\documentclass[10pt,journal]{IEEEtran}
\IEEEoverridecommandlockouts

\usepackage[utf8]{inputenc}
\usepackage[T1]{fontenc}
\usepackage{times}
\usepackage{graphicx}
\usepackage{amsmath,amssymb,amsfonts,bm,amsthm}
\usepackage{booktabs}
\usepackage{multirow}
\usepackage{enumitem}
\usepackage{xcolor}
\usepackage{url}
\usepackage[hidelinks]{hyperref}
\usepackage{caption}
\usepackage{subcaption}
\usepackage{algorithm}
\usepackage{algpseudocode}
\usepackage{array}
\usepackage{mathtools}
\usepackage{float}
\usepackage{dblfloatfix}
\usepackage{cleveref}
\usepackage{balance}
\usepackage{cite}
\newtheorem{theorem}{Theorem}

\newtheorem{assumption}{Assumption}

\title{EvoOMG: An Evolution-Oriented Multi-Agent Guidance Framework for Heterogeneous Legacy-and-MLO Wi-Fi Networks}

\author{Junjie~Wu, Lingjian~Zhou, Zerui~Shao, Yi~Zou, Tianrui~Li \IEEEmembership{Senior Member, IEEE}, \\ Yi~Zhang \IEEEmembership{Senior Member, IEEE}, Ziyuan~Yang \IEEEmembership{Member, IEEE}

\thanks{J. Wu, Z. Shao, and T, Li are with the School of Computing and Artificial Intelligence, Southwest Jiaotong University, Chengdu 611756, China. (e-mail: jjw@swjtu.edu.cn, zeruishao.zr@gmail.com, trli@swjtu.edu.cn).}
\thanks{L. Zhou and Y. Zou are with the School of Software, Nanchang Hangkong University, Nanchang 330063, China. (e-mail: 24201534@stu.nchu.edu.cn, 71410@nchu.edu.cn).}
\thanks{Y. Zhang is with the School of Cyber Science and Engineering, Sichuan University, Chengdu 610207, China. (e-mail: yzhang@scu.edu.cn).}
\thanks{Z. Yang is with the Nanyang Technological University, Singapore 639798, Singapore. (e-mail: cziyuanyang@gmail.com).}
\thanks{\emph{This work has been submitted to the IEEE for possible publication.  Copyright may be transferred without notice, after which this version may no longer be accessible.}}
}

\begin{document}
\maketitle

\begin{abstract}
The gradual deployment of Wi-Fi~7/8 multi-link operation (MLO) will lead to long-term coexistence between legacy non-MLO stations (STAs) and MLO-capable STAs in WLANs. This mixed deployment makes throughput optimization challenging because legacy STAs follow single-link contention and transmission, whereas MLO-capable STAs can exploit multiple links with richer access opportunities. Existing learning-based methods usually treat such networks as homogeneous systems and directly map the current observation to a complete MAC action, which cannot faithfully represent both legacy single-link and MLO multi-link behaviors. To address this issue, we propose EvoOMG, an evolution-oriented multi-agent guidance framework for heterogeneous legacy-and-MLO Wi-Fi networks. EvoOMG reformulates throughput optimization as a standard-constrained staged multi-agent decision problem. Each agent encodes recent channel, queue, contention, and transmission histories, first generates contention guidance, and then produces aggregation guidance conditioned on the preceding access stage and standard-specific feasibility constraints. This autoregressive design follows the Wi-Fi MAC order of ``contention before transmission'' while preserving distinct protocol behaviors of legacy and MLO-capable STAs. NS-3 evaluations show that EvoOMG improves scheduled goodput, convergence stability, and MLO link utilization over static enhanced distributed channel access (EDCA), one-step MADDPG, and independent-learning baselines, achieving substantial performance gains in representative mixed-standard scenarios.
\end{abstract}


\section{Introduction} 

\IEEEPARstart{W}{ireless} local area networks (WLANs) are entering a long coexistence period in which non-multi-link operation stations (STAs) and emerging multi-link operation (MLO)-capable Wi-Fi~7/8 stations operate in the same deployment.
This coexistence makes Wi-Fi throughput optimization substantially more challenging than in conventional homogeneous WLANs \cite{deng2020ieee,TMC_Jayabal}. Since different device types exhibit distinct transmission behaviors and channel-access opportunities, the optimization must account for both device heterogeneity and their asymmetric access to wireless resources \cite{jung2025modeling,tmc_Arthi}.

In homogeneous WLAN scenarios, STAs usually follow similar single-link access mechanisms and expose homogeneous medium access control (MAC) configuration spaces. Under this setting, existing methods can effectively optimize MAC-layer contention and scheduling policies, since the underlying optimization reduces to a relatively stationary resource allocation problem and is not affected by standard-induced action heterogeneity or multi-link feasibility constraints~\cite{wydmanski2021cw_drl,li2024reinwifi}.
However, the coexistence of legacy STAs~(non-MLO STAs) and MLO-capable STAs invalidates the assumptions underlying prior studies.
Legacy STAs usually operate on a single available link, whereas MLO-capable STAs can use multiple links and may follow different simultaneous transmit-and-receive (STR) or non-simultaneous transmit-and-receive (NSTR) constraints and traffic identifier (TID)-to-link mapping mapping rules \cite{wu2024twc_aggregation,galati2024commag_wifi8,shafin2025jsac_wifi8_p2p,wei2024arxiv_npca}. As a result, different STAs may have different feasible actions and channel-access behaviors. Therefore, they cannot be modeled in the same way as homogeneous STAs in conventional WLANs.

\begin{figure*}[!t]
    \centering
    \subfloat[Flat one-shot control under mixed standards.]{
        \includegraphics[width=0.47\linewidth]{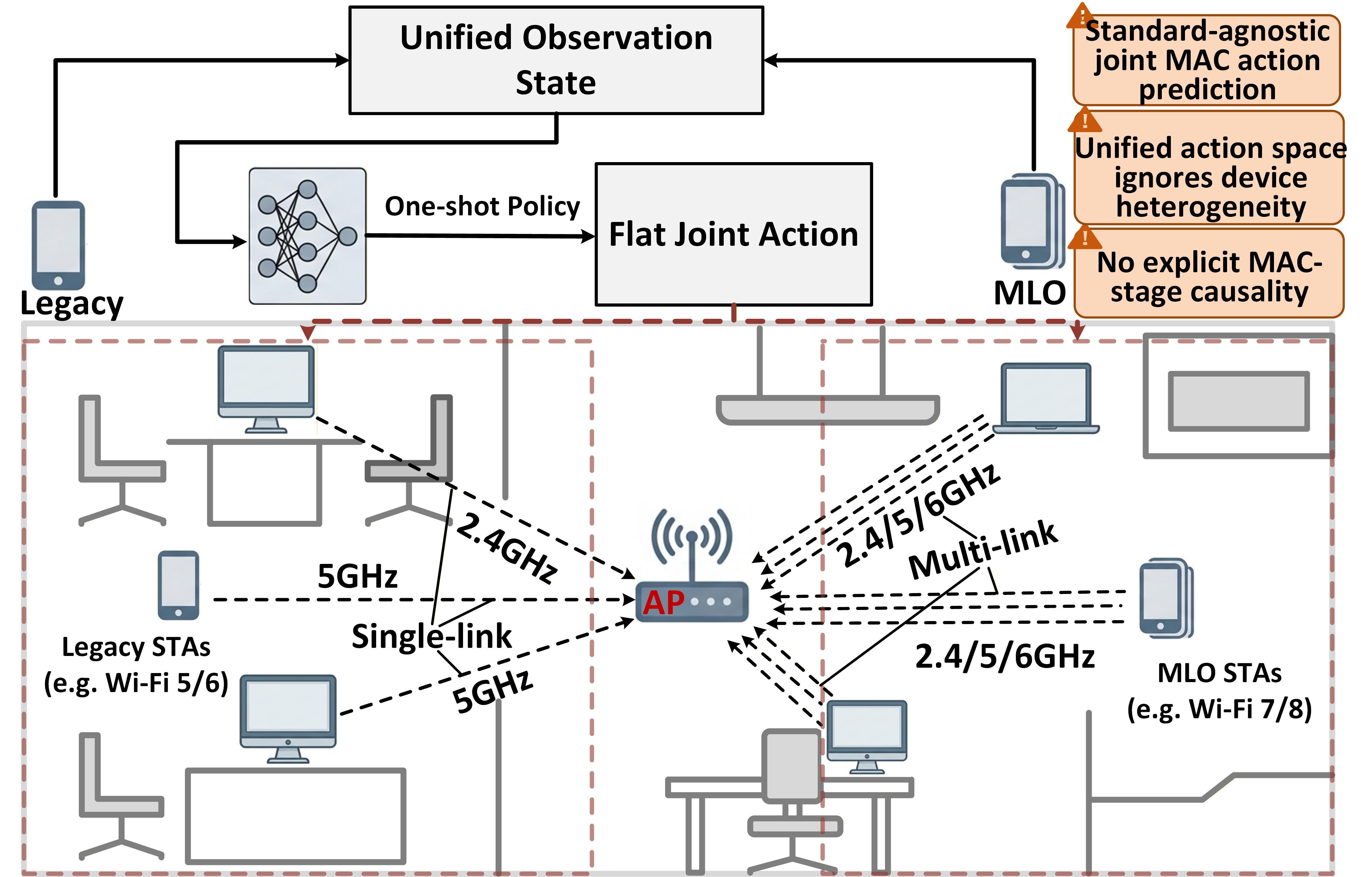}
        \label{fig:intro_bad}
    }
    \hfill
    \subfloat[Standard-aware autoregressive staged control.]{
        \includegraphics[width=0.47\linewidth]{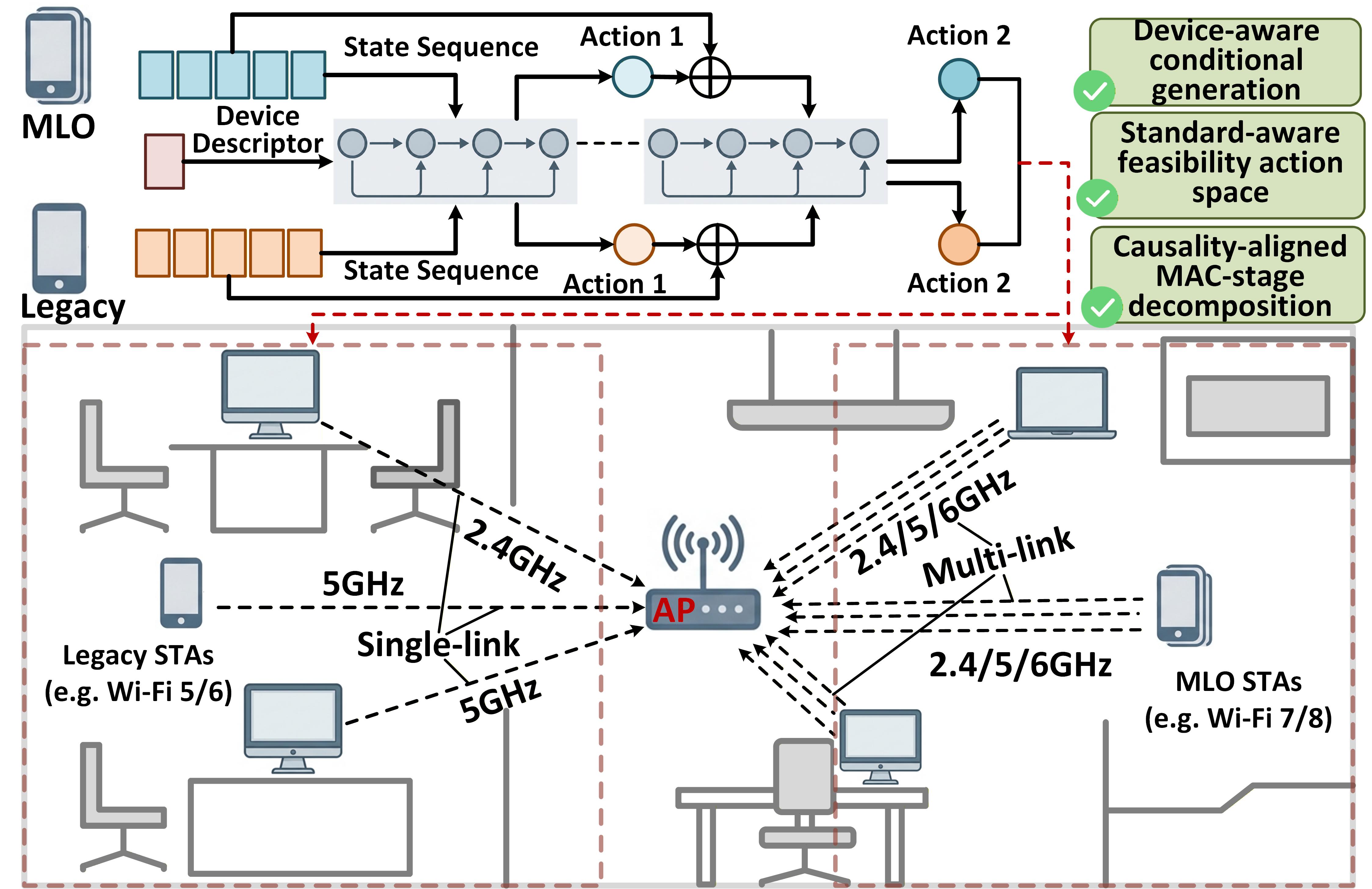}
        \label{fig:intro_good}
    }
    \caption{Motivation of EvoOMG. Conventional one-shot control treats legacy and MLO devices with a flat action space, whereas EvoOMG follows a standard-aware staged process that generates contention guidance before aggregation decisions.}
    \label{fig:intro_bad_good}
    \vspace{-0.5em}
\end{figure*}

This heterogeneity makes previous works difficult to extend in this challenging and practical scenario. 
Specifically, a policy that outputs all MAC-layer actions at once assumes that all STAs can be controlled in the same way. This is not true in a mixed legacy-and-MLO network, the same action may have different meanings for different device types. Moreover, Wi-Fi MAC control follows a natural order. Contention-window decisions first affect whether a STA can access the channel, while aggregation and link-scheduling decisions matter only after a transmission opportunity is obtained~\cite{lopez2022mlo_mwc,carrascosa2021mlo_latency,bellalta2022mlo_delay}. Therefore, heterogeneous Wi-Fi optimization should model both the differences among STA types and the ordered stages of MAC control, instead of treating all MAC actions as one flat decision.

Based on this observation, we revisit Wi-Fi optimization in mixed-standard environments from the perspective of closing the heterogeneity-induced performance gap via behavioral decomposition . The fundamental challenge arises from the coexistence of heterogeneous STA capabilities: MLO-capable STAs ~\cite{zhang2024iotj_mlo,avallone2025emlsr_jsac} can exploit multi-link transmission, STR/NSTR switching, and TID-to-link mapping, whereas legacy STAs are constrained to single-link contention-based access. This discrepancy induces a structural gap between the optimal performance envelope enabled by MLO and the achievable performance of learning-based MAC policies. We argue that this gap is not merely a function approximation issue, but fundamentally stems from a behavioral mismatch: existing methods attempt to directly learn a flat joint policy over contention, link scheduling, and aggregation decisions, without explicitly respecting either the standard-induced feasibility constraints or the inherent temporal ordering of the Wi-Fi MAC protocol. As a result, the high-dimensional heterogeneous action space is treated as a single-step mapping, which obscures the sequential dependency between contention, access, and transmission~\cite{wu_2025_iotj}.

To bridge this heterogeneity gap, we rethink the problem as a structured sequential decision process ~\cite{shen2025numerical_pruning,Han_2025_CVPR}, where the original high-dimensional joint action is decomposed into a series of conditional sub-decisions aligned with the MAC execution pipeline. Specifically, we reformulate the optimization problem from a flat policy learning task into a heterogeneity-aware staged control problem, in which contention control, link-level scheduling, and aggregation configuration are generated sequentially and conditioned on both device standard and prior decisions. This reformulation transforms the original coupled optimization into a structured decision hierarchy that explicitly aligns with protocol causality. This perspective leads to a central question:

 \emph{\textbf{How can we reformulate the high-dimensional, standard-constrained joint action optimization problem into a structured sequential decision process to bridge the performance gap, by explicitly decomposing MAC-layer control into conditional, causality-aligned sub-decisions?}}

To answer this question, we propose \textbf{EvoOMG}, an \underline{\textbf{Evo}}lution-\underline{\textbf{O}}riented \underline{\textbf{M}}ulti-agent \underline{\textbf{G}}uidance framework for heterogeneous legacy-and-MLO Wi-Fi networks. As shown in Fig.~\ref{fig:intro_bad_good}, EvoOMG differs from flat one-shot optimizers by encoding recent protocol histories, including channel, queue, contention, and transmission states, and then generating MAC guidance in a staged manner: contention guidance is produced first, while aggregation and link-related guidance are generated conditioned on the preceding access stage and standard-specific feasibility constraints. This design allows legacy STAs and MLO-capable STAs to be optimized within a unified MADRL framework while preserving their distinct protocol behaviors. EvoOMG follows the centralized-training and decentralized-execution paradigm~\cite{yu2022mappo}, and can optionally incorporate federated aggregation for distributed deployment under heterogeneous Wi-Fi domains~\cite{yang2025patient,YANG_FL1,wei2022low_latency_fl}. We implement EvoOMG in an NS-3-based heterogeneous Wi-Fi environment that exposes protocol-level observations such as SNR, PER, idle-time ratio, previous CW settings, aggregation decisions, and per-device throughput, enabling evaluation under realistic mixed legacy/MLO deployments with STR/NSTR constraints~\cite{wu2024twc_aggregation,zhang2024iotj_mlo}.
The main contributions are summarized as follows:
\begin{itemize}[leftmargin=1.2em]
    \item We identify \emph{standard-induced action heterogeneity} and \emph{MAC-stage temporality} as two key challenges in mixed legacy-and-MLO WLANs, and formulate the problem as a standard-constrained staged multi-agent decision process.

    \item We propose EvoOMG, a standard-aware autoregressive MADRL framework that encodes recent protocol histories and generates MAC actions following the contention-to-aggregation order. This design enables unified yet device-aware control for legacy single-link STAs and MLO multi-link STAs.

    \item We conduct NS-3-based evaluations under legacy/MLO deployments and compare EvoOMG with static EDCA, one-step MADDPG, and independent-learning baselines. The results show that EvoOMG improves scheduled goodput, convergence stability, and MLO link utilization while maintaining standard-compliant action generation.
\end{itemize}

\section{System Model and Problem Formulation}

We consider an infrastructure-based heterogeneous WLAN implemented on the NS-3 Wi-Fi stack, as illustrated in Fig.~\ref{fig:system_model}. One access point (AP) serves a set of STAs
$\mathcal{N}=\mathcal{N}_{\ell}\cup\mathcal{N}_{m}$,
where $\mathcal{N}_{\ell}$ denotes legacy STAs, i.e., non-MLO STAs in the considered mixed deployment, and $\mathcal{N}_{m}$ denotes MLO-capable STAs. Legacy STAs mainly operate with single-link access, whereas MLO-capable STAs can exploit multiple candidate links, such as 2.4~GHz, 5~GHz, and 6~GHz links~\cite{zou2025infocom_mlo_slo_delay,gao2025jsac_mlo_slo_latency}. Each STA $i$ is associated with a standard indicator $\tau_i$, where $\tau_i=0$ represents a legacy single-link STA and $\tau_i=1$ represents an MLO-capable STA. For a legacy STA $i\in\mathcal{N}_{\ell}$, the feasible link set is a singleton, denoted by $\mathcal{K}_i=\{0\}$. For an MLO-capable STA $j\in\mathcal{N}_{m}$, the feasible link set is $\mathcal{K}_j=\{1,\ldots,K_j\}$, where $K_j\geq 2$. The traffic is further differentiated by TIDs, which are mapped to access categories (ACs). Let
$\mathcal{A}=\{\mathrm{VO},\mathrm{VI},\mathrm{BE},\mathrm{BK}\}$
denote the set of ACs, corresponding to voice, video, best-effort, and background traffic \cite{wu2024twc_aggregation}. Different ACs may have different contention parameters, queue priorities, and throughput weights. Therefore, the considered heterogeneous Wi-Fi system is modeled at the STA--AC--link level.

\begin{figure}[!t] \centering \includegraphics[width=\linewidth]{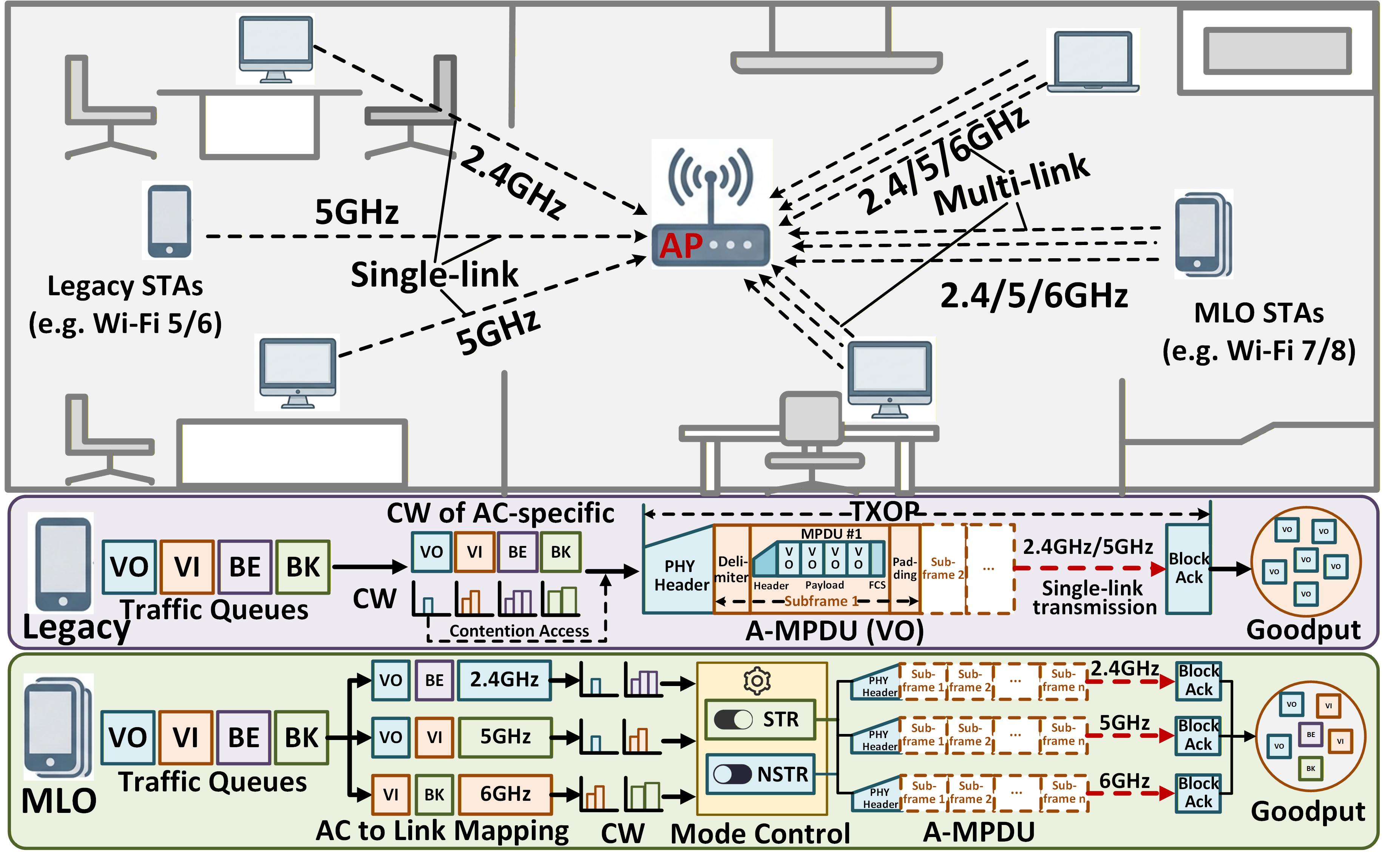} \caption{System model of the heterogeneous legacy-and-MLO Wi-Fi network. Legacy STAs use single-link access, whereas MLO-capable STAs schedule TID/AC traffic over feasible links under MLO constraints.} \label{fig:system_model} 
\vspace{-0.5em}
\end{figure}

\subsection{Channel, Mobility, and Link Model}

The considered heterogeneous WLAN operates in an indoor environment. The AP is fixed at position $\mathbf{p}_0$, while STAs move within a bounded service area according to the NS-3 \texttt{RandomWalk2dMobilityModel}. Let $\mathbf{p}_i(t)$ denote the position of STA $i$ at decision epoch $t$, and let $d_i(t)=\|\mathbf{p}_i(t)-\mathbf{p}_0\|$ be the AP--STA distance. The large-scale path loss is modeled as \cite{wu_tvt_2026}:
\begin{equation}
PL_i(t)=PL(d_0)+10\eta\log_{10}\left(\frac{d_i(t)}{d_0}\right)+X_{\sigma},
\label{eq:pathloss}
\end{equation}
where $d_0$ is the reference distance, $\eta$ is the path-loss exponent, and $X_{\sigma}$ denotes shadowing.

For a legacy STA $i\in\mathcal{N}_{\ell}$, the channel condition is defined on its single feasible link $k=0$. For an MLO-capable STA $j\in\mathcal{N}_{m}$, the channel condition is link-specific because candidate links may operate on different frequency bands and experience different propagation, interference, and contention conditions~\cite{lian2025jsac_mab_llm,jung2023iotj_otop}. For each feasible STA--link pair $(i,k)$, the NS-3 PHY/MAC stack provides the SNR, packet error ratio (PER), and channel idle-time ratio, denoted by $\gamma_{i,k}(t)$, $e_{i,k}(t)$, and $\iota_{i,k}(t)$, respectively.

The SNR $\gamma_{i,k}(t)$ characterizes PHY-layer channel quality, the PER $e_{i,k}(t)$ reflects transmission reliability after rate adaptation and retransmissions, and the idle-time ratio $\iota_{i,k}(t)$ captures MAC-layer contention and residual channel opportunity. For MLO-capable STAs, these statistics are collected separately over all candidate links $k\in\mathcal{K}_j$, enabling the controller to distinguish congested high-rate links from more reliable low-rate links. In contrast, legacy STAs do not have link-selection flexibility, and their performance is determined by the single-link channel condition, contention process, and TID/AC-specific queue behavior. This difference motivates the later STA--AC--link goodput model, where legacy goodput is computed over one link while MLO goodput is aggregated over active feasible links.

\subsection{TID-Aware Traffic Scheduling}

The heterogeneous WLAN carries traffic with different TIDs. Following the Wi-Fi QoS mechanism, TIDs are mapped into four ACs, denoted by
$\mathcal{A}=\{\mathrm{VO},\mathrm{VI},\mathrm{BE},\mathrm{BK}\}$,
corresponding to voice, video, best-effort, and background traffic, respectively. Let $\mathcal{U}_a$ denote the set of TIDs belonging to AC $a$. Recent studies on real-time industrial IoT and wireless time-sensitive networking also show that traffic differentiation and predictable latency are becoming increasingly important for Wi-Fi-based low-latency services~\cite{behnke2024iotj_realtime_iot}.

The considered system is modeled at the STA--AC--link level. For a legacy STA $i\in\mathcal{N}_{\ell}$, all AC queues are served over the single feasible link $k=0$. For an MLO-capable STA $j\in\mathcal{N}_{m}$, AC traffic can be scheduled over a candidate link set $\mathcal{K}_j$. Let $\rho_{i,a,k}(t)\in\{0,1\}$ denote whether AC $a$ traffic of STA $i$ is scheduled on link $k$ at epoch $t$. The legacy single-link scheduling constraint is $\rho_{i,a,0}(t)=1,\quad
i\in\mathcal{N}_{\ell},\; a\in\mathcal{A}$. For an MLO-capable STA $j\in\mathcal{N}_{m}$, AC traffic can be assigned to one or multiple feasible links according to the MLO capability, TID-to-link mapping, and link availability:
\begin{equation}
\rho_{j,a,k}(t)\in\{0,1\},\quad
j\in\mathcal{N}_{m},\; a\in\mathcal{A},\; k\in\mathcal{K}_j.
\label{eq:mlo-ac-link-schedule}
\end{equation}

If a certain TID/AC is not allowed to use link $k$ by the TID-to-link mapping or standard configuration, then $\rho_{j,a,k}(t)=0$.

\subsection{MAC Feasibility under Legacy and MLO}

The access behavior of each AC is governed by EDCA-like contention and aggregation procedures. For STA $i$, AC $a$, and link $k$, let $c_{i,a,k}(t)$ denote the contention-window-related configuration, and let $l_{i,a,k}(t)$ denote the aggregation frame length, e.g., the A-MPDU aggregation level. These two variables jointly affect the probability of obtaining a transmission opportunity and the amount of useful payload delivered once the channel is accessed. The feasible MAC configurations are standard-, AC-, and link-dependent:
\begin{equation}
c_{i,a,k}(t)\in\mathcal{W}_{i,a,k},\quad
l_{i,a,k}(t)\in\mathcal{L}_{i,a,k}.
\label{eq:cw-agg-feasible}
\end{equation}
where $\mathcal{W}_{i,a,k}$ and $\mathcal{L}_{i,a,k}$ denote the feasible contention-window and A-MPDU aggregation-length sets of STA $i$ for AC $a$ on link $k$, respectively. For legacy STAs, the feasible sets are defined only on the single link $k=0$. For MLO-capable STAs, the feasible sets are link-specific because different links may have different channel quality, contention intensity, and MLO availability.

To model MLO link activation, let $z_{j,k}(t)\in\{0,1\}$ indicate whether link $k$ of MLO-capable STA $j$ is active at epoch $t$. The AC-level scheduling indicator must satisfy
\begin{equation}
\rho_{j,a,k}(t)\leq z_{j,k}(t),
\quad j\in\mathcal{N}_{m},\; a\in\mathcal{A},\; k\in\mathcal{K}_j.
\label{eq:rho-z-coupling}
\end{equation}

Under STR, multiple links can be active concurrently if they are feasible, while practical performance still depends on cross-link interference, traffic steering, and link-access constraints~\cite{iturria2023icc_rsac}. Under NSTR, conflicting links cannot be active at the same time. Let $\mathcal{E}_j^{\mathrm{NSTR}}$ denote the set of conflicting link pairs of MLO-capable STA $j$. The NSTR constraint is:
\begin{equation}
z_{j,k}(t)+z_{j,k'}(t)\leq1,\quad
j\in\mathcal{N}_{m},\; (k,k')\in\mathcal{E}_j^{\mathrm{NSTR}}.
\label{eq:nstr-link-activation}
\end{equation}

This feasibility model captures the essential access difference between legacy and MLO-capable STAs. Legacy STAs perform AC-specific contention and aggregation on one link, whereas MLO-capable STAs must jointly consider AC-to-link scheduling, per-link contention behavior, aggregation length, and STR or NSTR-compliant link activation.

\subsection{Scheduled Goodput Model}
Since the system is implemented on the NS-3 Wi-Fi stack, packet transmission is not modeled by an ideal Shannon-rate expression. Instead, the delivered payload is generated by the protocol-level process, including EDCA contention, backoff, PHY rate control, frame aggregation, retransmission, acknowledgement or block acknowledgement, and MLO link scheduling. Therefore, we use scheduled goodput as the throughput metric.

Let $B^{\mathrm{succ}}_{i,a,k}(t)$ denote the successfully delivered application-layer payload bits of STA $i$ with AC $a$ over link $k$ during interval $\Delta t$. This quantity excludes MAC/PHY headers, control signaling, failed retransmissions, and acknowledgement overhead, while such signaling and retransmission costs are still reflected in the consumed channel time in NS-3. The STA-AC link goodput is defined as:
\begin{equation}
G_{i,a,k}(t)=\frac{B^{\mathrm{succ}}_{i,a,k}(t)}{\Delta t}.
\label{eq:ac-link-goodput}
\end{equation}

For a legacy STA $i\in\mathcal{N}_{\ell}$, all AC traffic is served on the single link. Thus, the AC-level goodput is:
\begin{equation}
G_{i,a}^{(\mathrm{leg})}(t)=G_{i,a,0}(t).
\label{eq:legacy-ac-goodput}
\end{equation}

For an MLO-capable STA $j\in\mathcal{N}_{m}$, the AC-level goodput is aggregated over the scheduled active links:
\begin{equation}
G_{j,a}^{(\mathrm{mlo})}(t)=
\sum_{k\in\mathcal{K}_j}
\rho_{j,a,k}(t)G_{j,a,k}(t).
\label{eq:mlo-ac-goodput}
\end{equation}

The goodput of each AC is affected by both contention and aggregation, which is consistent with recent studies showing that traffic allocation, channel selection, and MAC parameter tuning jointly determine MLO performance~\cite{lian2025jsac_mab_llm}. Hence, $G_{i,a,k}(t)$ is implicitly determined by the protocol-level interaction among $q_{i,a}(t)$, $c_{i,a,k}(t)$, $l_{i,a,k}(t)$, $\rho_{i,a,k}(t)$, the link condition, and the NS-3 process. The class-wise goodput of legacy and MLO-capable STAs is computed as:
\begin{equation}
\begin{split}
G_{\mathrm{leg}}(t)=
&\sum_{i\in\mathcal{N}_{\ell}}
\sum_{a\in\mathcal{A}}
w_a G_{i,a}^{(\mathrm{leg})}(t),\\
G_{\mathrm{mlo}}(t)=
&\sum_{j\in\mathcal{N}_{m}}
\sum_{a\in\mathcal{A}}
w_a G_{j,a}^{(\mathrm{mlo})}(t),
\end{split}
\label{eq:class-goodput}
\end{equation}
where $w_a$ denotes the AC-specific service weight. In general, $w_{\mathrm{VO}}$ and $w_{\mathrm{VI}}$ can be set larger than $w_{\mathrm{BE}}$ and $w_{\mathrm{BK}}$ to reflect the higher priority of delay-sensitive traffic. The total scheduled goodput of the heterogeneous Wi-Fi system is:
\begin{equation}
G_{\mathrm{sys}}(t)=G_{\mathrm{leg}}(t)+G_{\mathrm{mlo}}(t).
\label{eq:system-goodput}
\end{equation}

This goodput formulation captures TID/AC-level service differentiation, legacy single-link scheduling, MLO multi-link traffic distribution, and the impact of CW and A-MPDU configurations on protocol-level payload delivery.

\subsection{Problem Formulation}

The objective is to optimize the long-term scheduled goodput of the heterogeneous Wi-Fi system under standard-specific, AC-specific, and MLO-mode-specific feasibility constraints. Let $\pi$ denote the control policy that determines AC-level MAC configurations and link-scheduling decisions over time. Based on the scheduled goodput model in Eq.~\eqref{eq:system-goodput}, the instantaneous utility is defined as:
\begin{equation}
U(t)=G_{\mathrm{sys}}(t),
\label{eq:utility-goodput}
\end{equation}
where $G_{\mathrm{sys}}(t)$ is the total scheduled goodput of the heterogeneous legacy-and-MLO Wi-Fi system. This utility characterizes the performance under mixed legacy-and-MLO constraints and implicitly reflects the gap between feasible learned policies and the MLO-enabled performance envelope. The throughput-oriented optimization problem is formulated as:
\begin{subequations}
\label{eq:problem}
\begin{align}
\mathbf{P}:\quad
\max_{\pi}\quad
& \mathbb{E}_{\pi}\!\left[
\sum_{t=1}^{T} U(t)
\right] \label{eq:problem-obj}\\
\mathrm{s.t.}\quad
& \rho_{i,a,0}(t)=1,
\quad \forall i\in\mathcal{N}_{\ell},\, a\in\mathcal{A},
\label{eq:p-legacy-link}\\
& \rho_{j,a,k}(t)\in\{0,1\},
\quad \forall j\in\mathcal{N}_{m},\, a\in\mathcal{A},\, k\in\mathcal{K}_j,
\label{eq:p-mlo-rho}\\
& \rho_{j,a,k}(t)\leq z_{j,k}(t),
\quad \forall j\in\mathcal{N}_{m},\, a\in\mathcal{A},\, k\in\mathcal{K}_j,
\label{eq:p-rho-z}\\
& z_{j,k}(t)+z_{j,k'}(t)\leq1,
\quad \forall j\in\mathcal{N}_{m},\, (k,k')\in\mathcal{E}_j^{\mathrm{NSTR}},
\label{eq:p-nstr}\\
& c_{i,a,k}(t)\in\mathcal{W}_{i,a,k},
\quad \forall i\in\mathcal{N},\, a\in\mathcal{A},\, k\in\mathcal{K}_i,
\label{eq:p-cw}\\
& l_{i,a,k}(t)\in\mathcal{L}_{i,a,k},
\quad \forall i\in\mathcal{N},\, a\in\mathcal{A},\, k\in\mathcal{K}_i,
\label{eq:p-agg}\\
& \pi\in\Pi_{\mathrm{std}},
\label{eq:p-std}
\end{align}
\end{subequations}
where $T$ denotes the finite optimization horizon, i.e., the number of decision epochs considered in one control period. $\mathbb{E}_{\pi}[\cdot]$ denotes the expectation over trajectories induced by policy $\pi$ and the NS-3 protocol dynamics. Constraint~\eqref{eq:p-legacy-link} enforces the single-link access behavior of legacy STAs, where all AC traffic of a legacy STA is served on its only feasible link. 
Constraint~\eqref{eq:p-mlo-rho} defines the binary AC-to-link scheduling decision for MLO-capable STAs. 
Constraint~\eqref{eq:p-rho-z} couples AC-level scheduling with link activation, ensuring that AC traffic can be assigned to link $k$ only when this link is active. 
Constraint~\eqref{eq:p-nstr} imposes the NSTR feasibility requirement, where two conflicting links in $\mathcal{E}_j^{\mathrm{NSTR}}$ cannot be activated simultaneously. 
Constraint~\eqref{eq:p-cw} restricts the CW-related action to the feasible set $\mathcal{W}_{i,a,k}$ determined by the device standard, AC type, and link condition. 
Constraint~\eqref{eq:p-agg} restricts the A-MPDU aggregation-length action to the feasible set $\mathcal{L}_{i,a,k}$. 
Constraint~\eqref{eq:p-std} further requires the policy to belong to the standard-compliant policy space $\Pi_{\mathrm{std}}$, which excludes invalid MAC configurations and infeasible legacy/MLO actions.

\section{EvoOMG Design}

This section presents EvoOMG, an evolution-oriented multi-agent guidance framework for heterogeneous legacy-and-MLO Wi-Fi networks. Based on the STA--AC--link scheduled goodput formulation in Section~II, EvoOMG learns distributed MAC guidance policies that jointly consider TID/AC service differentiation, legacy single-link access, MLO multi-link scheduling, and the protocol-level coupling between contention-window configuration and aggregation-frame length. As shown in Fig.~\ref{fig:framework}, EvoOMG integrates a history-aware state encoder, a Transformer-based autoregressive actor, centralized critics for CTDE training, and an optional federated aggregation module for coordinating heterogeneous Wi-Fi domains.

\begin{figure*}[!t]
  \centering
  \includegraphics[width=\linewidth]{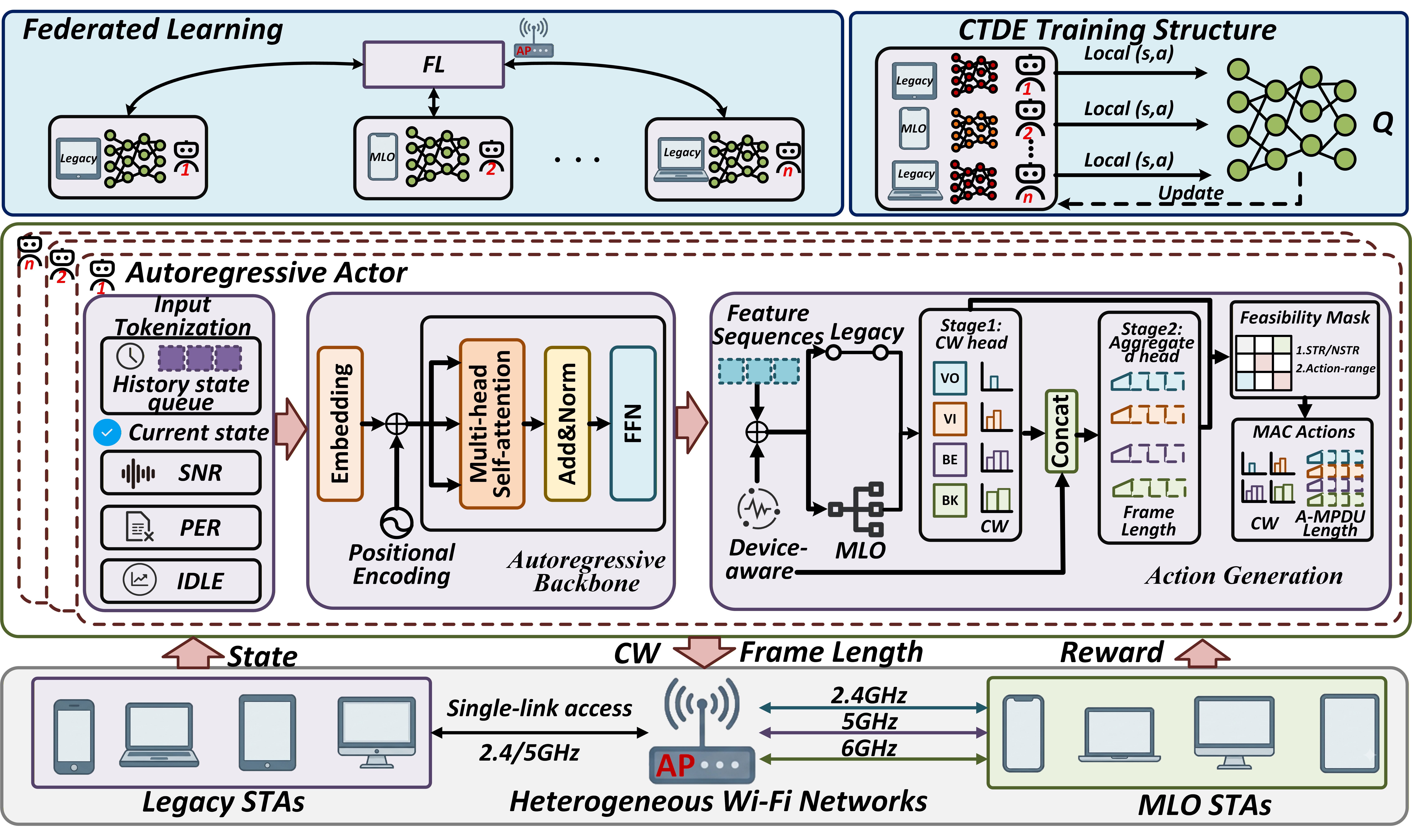}
  \caption{Overall architecture of EvoOMG. Each agent uses a transformer-based autoregressive actor to generate CW actions before A-MPDU length decisions. Centralized critics support CTDE-based multi-agent training, while optional federated aggregation enables distributed training across heterogeneous Wi-Fi domains.}
  \label{fig:framework}
  \vspace{-0.5em}
\end{figure*}

\subsection{Decentralized Partially Observable Markov Decision Process}

EvoOMG formulates heterogeneous Wi-Fi MAC guidance as a decentralized partially observable Markov decision process (Dec-POMDP):
\begin{equation}
\mathcal{G}_{\mathrm{D}}
=
\left(
\mathcal{N},
\mathcal{S},
\{\mathcal{O}_i\}_{i\in\mathcal{N}},
\{\mathcal{B}_i\}_{i\in\mathcal{N}},
P,
R,
\delta
\right),
\label{eq:decpomdp}
\end{equation}
where $\mathcal{N}$ is the STA set, $\mathcal{S}$ is the global network state, $\mathcal{O}_i$ and $\mathcal{B}_i$ are the local observation and action spaces of STA $i$, $P$ is the protocol-induced transition kernel, $R$ is the training reward, and $\delta\in(0,1)$ is the discount factor. 

\textbf{1) State.}
At decision epoch $t$, each STA observes its local STA--AC--link protocol state. For STA $i$, AC $a$, and feasible link $k\in\mathcal{K}_i$, the local observation is defined as:
\begin{equation}
\bm{o}_{i,a,k}(t)
=
\left[
\bm{\chi}_{i,k}(t),
q_{i,a}(t),
\bm{\psi}_{i,a,k}(t-1),
G_{i,a,k}(t-1)
\right],
\label{eq:compact-local-observation}
\end{equation}
where the link-state vector is  $\bm{\chi}_{i,k}(t)
=
\left[
\gamma_{i,k}(t),
e_{i,k}(t),
\iota_{i,k}(t)
\right]$,
and the previous MAC-action vector is $\bm{\psi}_{i,a,k}(t-1)
=
\left[
\rho_{i,a,k}(t-1),
z_{i,k}(t-1),
c_{i,a,k}(t-1),
l_{i,a,k}(t-1)
\right]$.
Here $\bm{\chi}_{i,k}(t)$ contains the SNR, PER, and idle-time ratio of link $k$; $q_{i,a}(t)$ is the AC-level queue backlog; $\bm{\psi}_{i,a,k}(t-1)$ records the previous scheduling, link activation, CW, and A-MPDU length; and $G_{i,a,k}(t-1)$ is the previous scheduled goodput.

To distinguish different Wi-Fi standards, each STA is associated with a device descriptor $\bm{d}_i=[\tau_i,\mu_i]$, where $\tau_i\in\{0,1\}$ indicates legacy or MLO capability, and $\mu_i\in\{\mathrm{SL},\mathrm{STR},\mathrm{NSTR}\}$ denotes the access mode. For legacy STAs, $\mu_i=\mathrm{SL}$ and $\mathcal{K}_i=\{0\}$. The local observation of STA $i$ is $\bm{o}_{i}(t)
=
\left\{
\bm{o}_{i,a,k}(t),
\forall a\in\mathcal{A},\;
k\in\mathcal{K}_i
\right\}
\cup
\{\bm{d}_i\}$.
To capture temporal MAC dynamics, EvoOMG maintains a sliding history:
\begin{equation}
\bm{O}_{i}(t)
=
\left[
\bm{o}_{i}(t-S+1),\ldots,\bm{o}_{i}(t)
\right],
\label{eq:sta-history}
\end{equation}
where $S$ is the history length. The history-aware decision feature is obtained by:
\begin{equation}
\bm{h}_i(t)
=
\mathrm{Enc}
\left(
\bm{O}_i(t),\bm{d}_i
\right),
\label{eq:history-encoder}
\end{equation}
where $\mathrm{Enc}(\cdot)$ denotes the Transformer-based history encoder. This history-aware encoding is aligned with recent AI-native Wi-Fi studies that emphasize environment-adaptive protocol control~\cite{wilhelmi2025commag_aiml_wifi,bellalta2024arxiv_aiml_radio}.

\textbf{2) Action.}
The local action is defined at the STA--AC--link level:
\begin{equation}
\bm{a}_i(t)
=
\left\{
\rho_{i,a,k}(t),
z_{i,k}(t),
c_{i,a,k}(t),
l_{i,a,k}(t)
\right\}_{a\in\mathcal{A},\,k\in\mathcal{K}_i}.
\label{eq:local-action-full}
\end{equation}
Here $\rho_{i,a,k}(t)$ is the AC-to-link scheduling indicator, $z_{i,k}(t)$ is the link-activation indicator, $c_{i,a,k}(t)$ is the CW-related action, and $l_{i,a,k}(t)$ is the A-MPDU aggregation-length action.

For a legacy STA $i\in\mathcal{N}_{\ell}$, the action degenerates to the single-link case:
\begin{equation}
\bm{a}_i^{(\ell)}(t)
=
\left\{
c_{i,a,0}(t),
l_{i,a,0}(t)
\right\}_{a\in\mathcal{A}},
\label{eq:legacy-action-design}
\end{equation}
with $\rho_{i,a,0}(t)=1$ and $z_{i,0}(t)=1$. For an MLO-capable STA $j\in\mathcal{N}_{m}$, the action preserves the multi-link structure:
\begin{equation}
\bm{a}_j^{(m)}(t)
=
\left\{
\rho_{j,a,k}(t),
z_{j,k}(t),
c_{j,a,k}(t),
l_{j,a,k}(t)
\right\}_{a\in\mathcal{A},\,k\in\mathcal{K}_j}.
\label{eq:mlo-action-design}
\end{equation}
The generated action is finally projected onto the standard-compliant feasible set. In the current NS-3 implementation, the learnable action vector consists of AC-specific CW configurations and the A-MPDU aggregation length. The AC-to-link scheduling variable $\rho_{i,a,k}(t)$ and link-activation variable $z_{i,k}(t)$ are determined by the standard-compliant feasibility projection, TID-to-link mapping, and STR/NSTR constraints rather than being independently optimized as unconstrained neural-network outputs.

\textbf{3) Objective and Reward.}
The optimization utility is the scheduled goodput of the heterogeneous Wi-Fi system:
\begin{equation}
U(t)=G_{\mathrm{sys}}(t).
\label{eq:decpomdp-utility}
\end{equation}

This utility is consistent with the problem formulation in Eq.~\eqref{eq:utility-goodput}. EvoOMG does not introduce additional explicit penalty terms for collision, delay, or packet dropping in the objective. These factors affect the delivered scheduled goodput through the NS-3 MAC/PHY process and are therefore reflected implicitly in $G_{\mathrm{sys}}(t)$. For stable neural-network training, the reward used by the autoregressive MADRL agent is a normalized version of the utility:
\begin{equation}
R(t)=\widetilde{U}(t)=\mathrm{Norm}\!\left(G_{\mathrm{sys}}(t)\right),
\label{eq:normalized-training-reward}
\end{equation}
where $\mathrm{Norm}(\cdot)$ denotes a normalization operator, such as min--max normalization or running-statistics normalization. This normalization only improves numerical stability and does not change the underlying optimization objective. The normalization is introduced solely for numerical stability in learning and does not alter the underlying objective of maximizing the heterogeneity-constrained system goodput.

\subsection{Autoregressive Staged Action Generator}

Different from prior learning-based MLO controllers that mainly focus on traffic steering or channel selection, EvoOMG further imposes the protocol order from contention guidance to aggregation guidance~\cite{iturria2023icc_rsac}. It factorizes the local policy into a contention stage and a transmission stage. For STA $i$, the factorized policy is written as:
\begin{equation}
\begin{aligned}
\pi_i(\bm{a}_i(t)\mid \bm{O}_i(t),\bm{d}_i)
&=
\pi_i^{\mathrm{cw}}
\left(
\bm{c}_i(t)\mid \bm{O}_i(t),\bm{d}_i
\right) \\
&\quad \times
\pi_i^{\mathrm{tr}}
\left(
\bm{u}_i(t)
\mid
\bm{c}_i(t),\bm{O}_i(t),\bm{d}_i
\right),
\end{aligned}
\label{eq:ar-policy-new}
\end{equation}
where
$\bm{c}_i(t)=\{c_{i,a,k}(t)\}_{a,k}$ is the contention-stage action and
$\bm{u}_i(t)=(\bm{\rho}_i(t),\bm{z}_i(t),\bm{l}_i(t))$ is the transmission-stage action, with
$\bm{\rho}_i(t)=\{\rho_{i,a,k}(t)\}_{a,k}$,
$\bm{z}_i(t)=\{z_{i,k}(t)\}_{k}$, and
$\bm{l}_i(t)=\{l_{i,a,k}(t)\}_{a,k}$.
The first stage generates AC-link contention guidance:
\begin{equation}
\hat{\bm{c}}_i(t)
=
f_{\theta_c}
\left(
\bm{h}_i(t),\bm{d}_i
\right),
\label{eq:cw-generation}
\end{equation}
where $\theta_c$ demotes contention head parameters. The contention guidance affects backoff behavior and channel access probability. The second stage generates aggregation guidance conditioned on the contention output:
\begin{equation}
\hat{\bm{l}}_i(t)
=
f_{\theta_l}
\left(
\bm{h}_i(t),
\bm{d}_i,
\hat{\bm{c}}_i(t)
\right),
\label{eq:aggregation-generation}
\end{equation}
where $\theta_l$ demotes aggregation head parameters.
The associated scheduling and link-activation variables are then obtained through the standard-compliant feasibility projection according to device capability, TID-to-link mapping, and STR/NSTR constraints.

This staged factorization explicitly aligns with the heterogeneity-aware causality structure of Wi-Fi MAC operations, where contention decisions determine access feasibility and subsequent aggregation decisions depend on both access outcome and device-specific constraints. The contention-stage decision determines access opportunity and collision risk, while the transmission-stage decision determines how much payload is aggregated and how feasible links are used after the access condition is formed. To ensure protocol feasibility, let
$\hat{\bm{a}}_i(t)=(\hat{\bm{c}}_i(t),\hat{\bm{l}}_i(t))$
denote the raw learnable actor output. The final MAC action is obtained by $\bm{a}_i(t)
=
\mathcal{P}_{\Omega_i}
\left(
\hat{\bm{a}}_i(t)
\right)$,
where $\Omega_i$ denotes the feasible action set determined by the device standard, AC type, TID-to-link mapping, and STR/NSTR mode. For legacy STAs, invalid MLO link actions are masked out and only single-link CW and aggregation actions are retained. For MLO-capable STAs, AC-to-link scheduling and link activation are constrained by MLO feasibility and NSTR conflicting-link constraints.

\subsection{Standard-Aware Multi-Agent Policy Heads}

To address the heterogeneous action spaces of legacy and MLO-capable STAs, EvoOMG uses standard-aware policy heads. The sequence encoder and early actor layers are shared across agents to capture common patterns, while the output heads are adapted according to the device descriptor $\bm{d}_i$.

For a legacy STA, the policy head only outputs single-link AC-specific CW and aggregation guidance. Its feasible action space is denoted by $\Omega_i^{(\mathrm{leg})}$, which enforces $\rho_{i,a,0}(t)=1$ and $z_{i,0}(t)=1$. For an MLO-capable STA, the feasible action space is denoted by $\Omega_i^{(\mathrm{mlo})}$, which further depends on STR or NSTR operation and determines valid AC-to-link scheduling and link activation after projection. The final policy output is represented as:
\begin{equation}
\bm{a}_i(t)
=
\begin{cases}
\mathcal{P}_{\Omega_i^{(\mathrm{leg})}}\!\left(\hat{\bm{a}}_i(t)\right),
& i\in\mathcal{N}_{\ell},\\[1mm]
\mathcal{P}_{\Omega_i^{(\mathrm{mlo})}}\!\left(\hat{\bm{a}}_i(t)\right),
& i\in\mathcal{N}_{m},
\end{cases}
\label{eq:standard-aware-policy}
\end{equation}
where $\hat{\bm{a}}_i(t)$ is the raw actor output, $\mathcal{P}_{\Omega}(\cdot)$ denotes the projection operator onto feasible action set $\Omega$, $\Omega_i^{(\mathrm{leg})}$ and $\Omega_i^{(\mathrm{mlo})}$ denote the feasible action sets of legacy and MLO-capable STAs, respectively, and the superscripts $(\mathrm{leg})$ and $(\mathrm{mlo})$ represent the legacy single-link case and the MLO-capable multi-link case. This design allows EvoOMG to share useful protocol knowledge across heterogeneous STAs while avoiding invalid actions caused by mixing legacy single-link and MLO multi-link behaviors in one flat output space.

\subsection{Centralized Critic and Federated Training Procedure}

EvoOMG follows the centralized-training and decentralized-execution (CTDE) paradigm within each local Wi-Fi domain. 
During local training, the centralized critic observes the joint histories and joint actions of all STA agents in the local heterogeneous WLAN to estimate the global action value. 
During execution, each STA agent generates its MAC action using only its local observation history and device descriptor. 
To support distributed deployment across heterogeneous Wi-Fi domains, EvoOMG can further adopt an optional federated-assisted training procedure.

Let $v$ index a local training client, e.g., an AP-side Wi-Fi domain or an NS-3 training environment, and let $\mathcal{N}^{v}$ denote the STA set in client $v$. 
The joint history and joint action in client $v$ are defined as:
\begin{equation}
\bm{O}^{v}(t)=\{\bm{O}^{v}_i(t)\}_{i\in\mathcal{N}^{v}},
\qquad
\bm{a}^{v}(t)=\{\bm{a}^{v}_i(t)\}_{i\in\mathcal{N}^{v}},
\label{eq:local-joint-history-action}
\end{equation}
where $\bm{O}^{v}_i(t)$ and $\bm{a}^{v}_i(t)$ are the local history and action of STA $i$ in client $v$, respectively. 
The local centralized critic is written as $Q_{\phi_v}
\left(
\bm{O}^{v}(t),
\bm{a}^{v}(t)
\right)$, where $\phi_v$ denotes the critic parameters of client $v$. The local replay buffer of client $v$ stores sequence transitions:
\begin{equation}
\mathcal{D}_v
=
\left\{
\left(
\bm{O}^{v}(t),
\bm{a}^{v}(t),
R^{v}(t),
\bm{O}^{v}(t+1)
\right)
\right\}.
\label{eq:local-replay-buffer}
\end{equation}

Different from one-step replay, $\mathcal{D}_v$ stores history tensors with the form $[\mathrm{batch},\mathrm{seq\_len},\mathrm{state\_dim}]$, so that local temporal protocol continuity is preserved.

For each sampled transition, the temporal-difference target is computed as:
\begin{equation}
y^{v}(t)
=
R^{v}(t)
+
\gamma
Q_{\phi_v^-}
\left(
\bm{O}^{v}(t+1),
\bm{a}^{v,-}(t+1)
\right),
\label{eq:local-td-target}
\end{equation}
where $\gamma$ is the discount factor, $Q_{\phi_v^-}$ is the local target critic, and $\bm{a}^{v,-}(t+1)$ is generated by the local target actors. 
The local critic loss is:
\begin{equation}
\mathcal{L}_{Q}^{v}(\phi_v)
=
\mathbb{E}_{\mathcal{D}_v}
\left[
\left(
Q_{\phi_v}(\bm{O}^{v}(t),\bm{a}^{v}(t))-y^{v}(t)
\right)^2
\right].
\label{eq:local-critic-loss}
\end{equation}

The actor parameters of STA $i$ in client $v$ are updated by deterministic policy gradient:
\begin{equation}
\nabla_{\theta_i^v}J_i^v
=
\mathbb{E}_{\mathcal{D}_v}
\left[
\nabla_{\theta_i^v}\pi_{\theta_i^v}(\bm{O}_i^v(t),\bm{\theta}_i^v)
\nabla_{\bm{a}_i^v}
Q_{\phi_v}(\bm{O}^{v}(t),\bm{a}^{v}(t))
\right],
\label{eq:local-actor-gradient}
\end{equation}
where $\theta_i^v$ denotes the actor parameters of STA $i$ in client $v$, and $\bm{\theta}_i^v$ is its device descriptor.

To support distributed deployment across heterogeneous Wi-Fi domains, EvoOMG can optionally adopt a lightweight federated aggregation mechanism. 
Let $\mathbf{w}_v^{(r)}$ denote the trainable model parameters of client $v$ after local training in communication round $r$, including the sequence encoder, staged actor, standard-aware policy heads, and critic parameters. 
The server aggregates the uploaded model parameters as:
\begin{equation}
\mathbf{w}^{(r+1)}
=
\sum_{v\in\mathcal{M}}
\alpha_v \mathbf{w}_v^{(r+1)},
\qquad
\sum_{v\in\mathcal{M}}\alpha_v=1,
\label{eq:federated-aggregation}
\end{equation}
where $\mathcal{M}$ denotes the set of participating local Wi-Fi clients and $\alpha_v$ is the aggregation weight. 
Each client then downloads $\mathbf{w}^{(r+1)}$ to initialize the next local training round. 
Since different clients may contain different legacy/MLO ratios, traffic mixtures, and STR/NSTR modes, the standard-aware policy heads and feasibility masks are retained during local training to reduce cross-standard drift. 
In this paper, federated aggregation is used as an optional deployment mechanism rather than the main algorithmic contribution.

\subsection{Theoretical Analysis}


This section analyzes the performance gap induced by heterogeneous MAC constraints in mixed legacy-and-MLO Wi-Fi networks. Let the local decision state of STA $i$ be denoted by $\bm{x}_i(t)
=
\left(
\bm{O}_i(t),
\bm{d}_i
\right),$ where $\bm{O}_i(t)$ is the local protocol history and $\bm{d}_i$ is the device descriptor. For STA $i$, define the contention-stage action as $\bm{c}_i(t)=\{c_{i,a,k}(t)\}_{a,k}$ and the transmission-stage action as $\bm{u}_i(t)
=
\left(
\bm{\rho}_i(t),
\bm{z}_i(t),
\bm{l}_i(t)
\right)$,
where $\bm{\rho}_i(t)=\{\rho_{i,a,k}(t)\}_{a,k}$, $\bm{z}_i(t)=\{z_{i,k}(t)\}_{k}$, and $\bm{l}_i(t)=\{l_{i,a,k}(t)\}_{a,k}$. The scheduled goodput contribution of STA $i$ can be abstracted as a stage-coupled function:
\begin{equation}
G_i(t)
=
A_i\!\left(
\bm{c}_i(t);
\bm{x}_i(t)
\right)
B_i\!\left(
\bm{u}_i(t);
\bm{c}_i(t),
\bm{x}_i(t)
\right),
\label{eq:stage-coupled-goodput}
\end{equation}
where $A_i(\cdot)$ denotes the access-efficiency term induced by AC-level contention and channel acquisition, and $B_i(\cdot)$ denotes the payload-efficiency term induced by AC-to-link scheduling, link activation, A-MPDU length, and STR/NSTR feasibility constraints. Consistent with the scheduled goodput definition in Eq.~\eqref{eq:system-goodput}, the discounted goodput objective is defined as:
\begin{equation}
J(\pi)=
\mathbb{E}_{\pi}
\left[
\sum_{t=0}^{\infty}
\delta^t G_{\mathrm{sys}}(t)
\right],
\quad \delta\in(0,1),
\label{eq:discounted-objective-new}
\end{equation}
where $\delta$ is the discount factor. This objective implicitly measures the gap between the achievable performance of the learned policy and the optimal STR-enabled performance envelope under standard heterogeneity.


\begin{theorem}[STR/NSTR Feasible-Set Envelope]
Assume that all scheduled goodput contributions are nonnegative and that the NSTR feasible link-activation set is a subset of the STR feasible link-activation set. Let $J_{\mathrm{NSTR}}^{\star}$ and $J_{\mathrm{STR}}^{\star}$ denote the optimal discounted scheduled-goodput objectives under NSTR-constrained and STR-enabled MLO operation, respectively. Then
\begin{equation}
J_{\mathrm{NSTR}}^{\star}
\leq
J_{\mathrm{STR}}^{\star}.
\label{eq:envelope-main-new}
\end{equation}
\end{theorem}

\begin{proof}
The NSTR mode restricts simultaneous activation of conflicting links, while STR allows a superset of feasible link-activation patterns. Since each scheduled goodput contribution is nonnegative, enlarging the feasible link-activation set cannot decrease the optimal achievable discounted goodput. The result follows from feasible-set inclusion.
\end{proof}


\begin{theorem}[Staged Optimal Factorization]
Let $Q^{\star}(\bm{x},\bm{c},\bm{u})$ denote the optimal action-value function, where $\bm{x}$ is the local decision state, $\bm{c}$ is the contention-stage action, and $\bm{u}$ is the transmission-stage action. Then there exists an optimal staged policy of the following form:
\begin{equation}
\pi^{\star}(\bm{c},\bm{u}\mid\bm{x})
=
\pi_c^{\star}(\bm{c}\mid\bm{x})
\pi_u^{\star}(\bm{u}\mid\bm{c},\bm{x}).
\label{eq:staged-factorization-new}
\end{equation}
\end{theorem}


\begin{proof}
For any fixed decision state $\bm{x}$, the optimal joint action solves
\begin{equation}
(\bm{c}^{\star}(\bm{x}),\bm{u}^{\star}(\bm{x}))
\in
\arg\max_{\bm{c},\bm{u}}
Q^{\star}(\bm{x},\bm{c},\bm{u}).
\label{eq:joint-max-proof}
\end{equation}
The joint maximization in Eq.~\eqref{eq:joint-max-proof} can be equivalently written as the following nested maximization:
\begin{equation}
\max_{\bm{c},\bm{u}}
Q^{\star}(\bm{x},\bm{c},\bm{u})
=
\max_{\bm{c}}
\left[
\max_{\bm{u}}
Q^{\star}(\bm{x},\bm{c},\bm{u})
\right].
\label{eq:nested-max-proof}
\end{equation}
Therefore, one optimal contention-stage action can be selected as
\begin{equation}
\bm{c}^{\star}(\bm{x})
\in
\arg\max_{\bm{c}}
\left[
\max_{\bm{u}}
Q^{\star}(\bm{x},\bm{c},\bm{u})
\right].
\end{equation}
For a given contention-stage action $\bm{c}$, the corresponding optimal transmission-stage action is
\begin{equation}
\bm{u}^{\star}(\bm{x},\bm{c})
\in
\arg\max_{\bm{u}}
Q^{\star}(\bm{x},\bm{c},\bm{u}).
\end{equation}
Thus, the optimal joint decision can be represented by first selecting $\bm{c}^{\star}(\bm{x})$ and then selecting $\bm{u}^{\star}(\bm{x},\bm{c})$ conditioned on the contention-stage action and the decision state.

Equivalently, the optimal policy can be written in staged form by defining
$\pi_c^{\star}(\bm{c}\mid\bm{x})$ over the optimal contention-stage actions and
$\pi_u^{\star}(\bm{u}\mid\bm{c},\bm{x})$ over the corresponding optimal transmission-stage actions. This gives
\begin{equation}
\pi^{\star}(\bm{c},\bm{u}\mid\bm{x})
=
\pi_c^{\star}(\bm{c}\mid\bm{x})
\pi_u^{\star}(\bm{u}\mid\bm{c},\bm{x}),
\end{equation}
which proves Eq.~\eqref{eq:staged-factorization-new}.
\end{proof}

The above theorem shows that the staged policy adopted by EvoOMG is not merely a heuristic decomposition. Under the exact action-value function, contention-first and transmission-conditioned generation is lossless relative to joint action optimization. In practical learning, this factorization also reduces the burden of learning asymmetric MAC dependencies implicitly from one flat action vector.

\begin{assumption}[Lipschitz Continuity]
The optimal action-value function is Lipschitz continuous with respect to the contention action, transmission action, and decision state. Namely, there exist constants $L_c$, $L_u$, and $L_x$ such that, for any feasible contention-stage actions $\bm{c}$ and $\bm{c}'$, transmission-stage actions $\bm{u}$ and $\bm{u}'$, and decision states $\bm{x}$ and $\bm{x}'$,
\begin{equation}
\left|
Q^{\star}(\bm{x},\bm{c},\bm{u})
-
Q^{\star}(\bm{x},\bm{c}',\bm{u})
\right|
\leq
L_c\|\bm{c}-\bm{c}'\|,
\end{equation}
\begin{equation}
\left|
Q^{\star}(\bm{x},\bm{c},\bm{u})
-
Q^{\star}(\bm{x},\bm{c},\bm{u}')
\right|
\leq
L_u\|\bm{u}-\bm{u}'\|,
\end{equation}
and
\begin{equation}
\left|
Q^{\star}(\bm{x},\bm{c},\bm{u})
-
Q^{\star}(\bm{x}',\bm{c},\bm{u})
\right|
\leq
L_x\|\bm{x}-\bm{x}'\|.
\end{equation}
\end{assumption}


\begin{assumption}[Stage-wise Approximation Errors]
Let $\bm{x}$ denote the ideal local decision state, and let $\hat{\bm{x}}$ denote the encoded decision-state representation generated by the history encoder. Given $\hat{\bm{x}}$, EvoOMG produces the contention-stage action $\hat{\bm{c}}(\hat{\bm{x}})$ and the transmission-stage action $\hat{\bm{u}}(\hat{\bm{x}},\hat{\bm{c}})$. Their expected approximation errors relative to the optimal staged actions satisfy
\begin{equation}
\mathbb{E}\!\left[
\|\hat{\bm{c}}(\hat{\bm{x}})-\bm{c}^{\star}(\bm{x})\|
\right]
\leq
\varepsilon_c,
\end{equation}
\begin{equation}
\mathbb{E}\!\left[
\|\hat{\bm{u}}(\hat{\bm{x}},\hat{\bm{c}})
-
\bm{u}^{\star}(\bm{x},\bm{c}^{\star}(\bm{x}))\|
\right]
\leq
\varepsilon_u,
\end{equation}
and the decision-state representation mismatch satisfies
\begin{equation}
\mathbb{E}\!\left[
\|\hat{\bm{x}}-\bm{x}\|
\right]
\leq
\varepsilon_x.
\end{equation}
\end{assumption}


\begin{theorem}[Performance Gap Bound]
Under Assumptions~1--2, let $\delta\in(0,1)$ denote the discount factor and define
\begin{equation}
\varepsilon_Q
=
L_c\varepsilon_c
+
L_u\varepsilon_u
+
L_x\varepsilon_x .
\end{equation}
Then the performance gap between the optimal policy and EvoOMG is bounded as follows:
\begin{equation}
J(\pi^{\star})-J(\pi_{\mathrm{EvoOMG}})
\leq
\frac{2\delta}{(1-\delta)^2}
\varepsilon_Q .
\label{eq:main-gap-bound-new}
\end{equation}
Equivalently,
\begin{equation}
J(\pi^{\star})-J(\pi_{\mathrm{EvoOMG}})
\leq
\frac{2\delta}{(1-\delta)^2}
\left(
L_c\varepsilon_c
+
L_u\varepsilon_u
+
L_x\varepsilon_x
\right).
\end{equation}
\end{theorem}

\begin{proof}
For a fixed decision state $\bm{x}$, let
$\bm{c}^{\star}(\bm{x})$ and
$\bm{u}^{\star}(\bm{x},\bm{c}^{\star}(\bm{x}))$
denote the optimal staged actions. EvoOMG observes the encoded decision-state representation $\hat{\bm{x}}$ and produces
$\hat{\bm{c}}(\hat{\bm{x}})$ and
$\hat{\bm{u}}(\hat{\bm{x}},\hat{\bm{c}})$.
We define the following quantity $\Delta_Q$ to characterize the value-function discrepancy induced by representation and policy approximation:
\begin{equation}
\Delta_Q
=
\left|
Q^{\star}\!\left(
\bm{x},
\bm{c}^{\star}(\bm{x}),
\bm{u}^{\star}(\bm{x},\bm{c}^{\star}(\bm{x}))
\right)
-
Q^{\star}\!\left(
\hat{\bm{x}},
\hat{\bm{c}}(\hat{\bm{x}}),
\hat{\bm{u}}(\hat{\bm{x}},\hat{\bm{c}})
\right)
\right|.
\label{eq:delta-q-def}
\end{equation}
Then, by adding and subtracting intermediate terms and using the Lipschitz continuity of $Q^{\star}$, we have
\begin{equation}
\begin{split}
\Delta_Q
\leq\;&
L_c
\left\|
\hat{\bm{c}}(\hat{\bm{x}})
-
\bm{c}^{\star}(\bm{x})
\right\| \\
&+
L_u
\left\|
\hat{\bm{u}}(\hat{\bm{x}},\hat{\bm{c}})
-
\bm{u}^{\star}(\bm{x},\bm{c}^{\star}(\bm{x}))
\right\| \\
&+
L_x
\left\|
\hat{\bm{x}}-\bm{x}
\right\|.
\end{split}
\label{eq:q-gap-bound}
\end{equation}
Taking expectation on both sides of Eq.~\eqref{eq:q-gap-bound} and applying Assumption~2 gives
\begin{equation}
\mathbb{E}
\left[
\Delta_Q
\right]
\leq
L_c\varepsilon_c
+
L_u\varepsilon_u
+
L_x\varepsilon_x
=
\varepsilon_Q .
\label{eq:expected-q-gap}
\end{equation}
Thus, $\varepsilon_Q$ upper-bounds the expected action-value degradation caused by three sources: contention-stage approximation, transmission-stage approximation, and decision-state representation mismatch.

By the standard discounted performance-difference bound for approximate policy improvement, an expected action-value degradation bounded by $\varepsilon_Q$ leads to
\begin{equation}
J(\pi^{\star})-J(\pi_{\mathrm{EvoOMG}})
\leq
\frac{2\delta}{(1-\delta)^2}
\varepsilon_Q .
\end{equation}
Substituting the definition of $\varepsilon_Q$ completes the proof.
\end{proof}



The bound in Eq.~\eqref{eq:main-gap-bound-new} clarifies the role of each EvoOMG component. The sequence encoder reduces $\varepsilon_x$ by preserving protocol history; the staged generator reduces $\varepsilon_c$ and $\varepsilon_u$ by explicitly modeling the contention-to-transmission dependency; and the standard-aware policy heads reduce approximation mismatch caused by heterogeneous legacy and MLO action spaces. Therefore, the theoretical analysis supports the main design principle of EvoOMG: heterogeneous Wi-Fi MAC control should be parameterized as a standard-aware staged decision process rather than a flat one-shot action mapping.


\section{Experimental Results}
\label{sec:experiments}


This section evaluates EvoOMG in an NS-3-based heterogeneous Wi-Fi environment with ns3-ai interaction. Following recent protocol-level MLO studies~\cite{zou2025infocom_mlo_slo_delay,gao2025jsac_mlo_slo_latency}, we consider one IEEE 802.11be AP serving both MLO STAs and IEEE 802.11ax legacy STAs under downlink user datagram protocol (UDP) traffic. The key system and learning settings are summarized in Table~\ref{tab:simulation_parameters}~\cite{wu_2025_iotj,wu2024twc_aggregation,wu_tvt_2026}. The evaluation covers training stability, legacy/MLO heterogeneity, TID/AC-aware traffic differentiation, and robustness under MLO PHY constraints, with MADDPG, IDDPG, conservative, and greedy policies as baselines.

\begin{table}[!t]
\caption{Key Simulation Parameters}
\label{tab:simulation_parameters}
\centering
\renewcommand{\arraystretch}{1.08}
\setlength{\tabcolsep}{4.5pt}
\footnotesize
\begin{tabular}{@{}p{0.56\linewidth}p{0.34\linewidth}@{}}
\toprule
Parameter & Value \\
\midrule
\multicolumn{2}{@{}l}{\textit{Wi-Fi system configuration}} \\
\midrule
Simulator & NS-3 + ns3-ai \\
AP / MLO / Legacy & 802.11be (Wi-Fi~7) / 802.11be (Wi-Fi~7) / 802.11ax (Wi-Fi~6) \\
MLO bands & 2.4 / 5 / 6 GHz \\
MLO bandwidths & 40 / 80 / 320 MHz \\
Legacy band / bandwidth & 5 GHz / 80 MHz \\
Propagation / mobility & Log-distance / RandomWalk2d \\
Traffic type & Downlink UDP \\
Offered load & 900 / 800 Mbit/s \\
Packet size & 1440 bytes \\
CW range & $[3,1023]$ \\
Aggregation range & $[1,4096]$ \\
Interaction interval & 500 ms \\
Warm-up / measurement time & 10 s / 60 s \\
\midrule
\multicolumn{2}{@{}l}{\textit{Learning configuration}} \\
\midrule
History length & 5 \\
Transformer dimension / heads & 16 / 2 \\
Actor / critic learning rates & $3\times10^{-4}$ / $10^{-3}$ \\
Replay buffer / batch size & $10^5$ / 64 \\
Discount factor & 0.99 \\
Soft update coefficient & 0.005 \\
\bottomrule
\end{tabular}
\end{table}


\subsection{Training Convergence and Critic Stability}

\begin{figure}[!t]
    \centering
    \begin{subfigure}[t]{.48\linewidth}
        \centering
        \includegraphics[width=\linewidth]{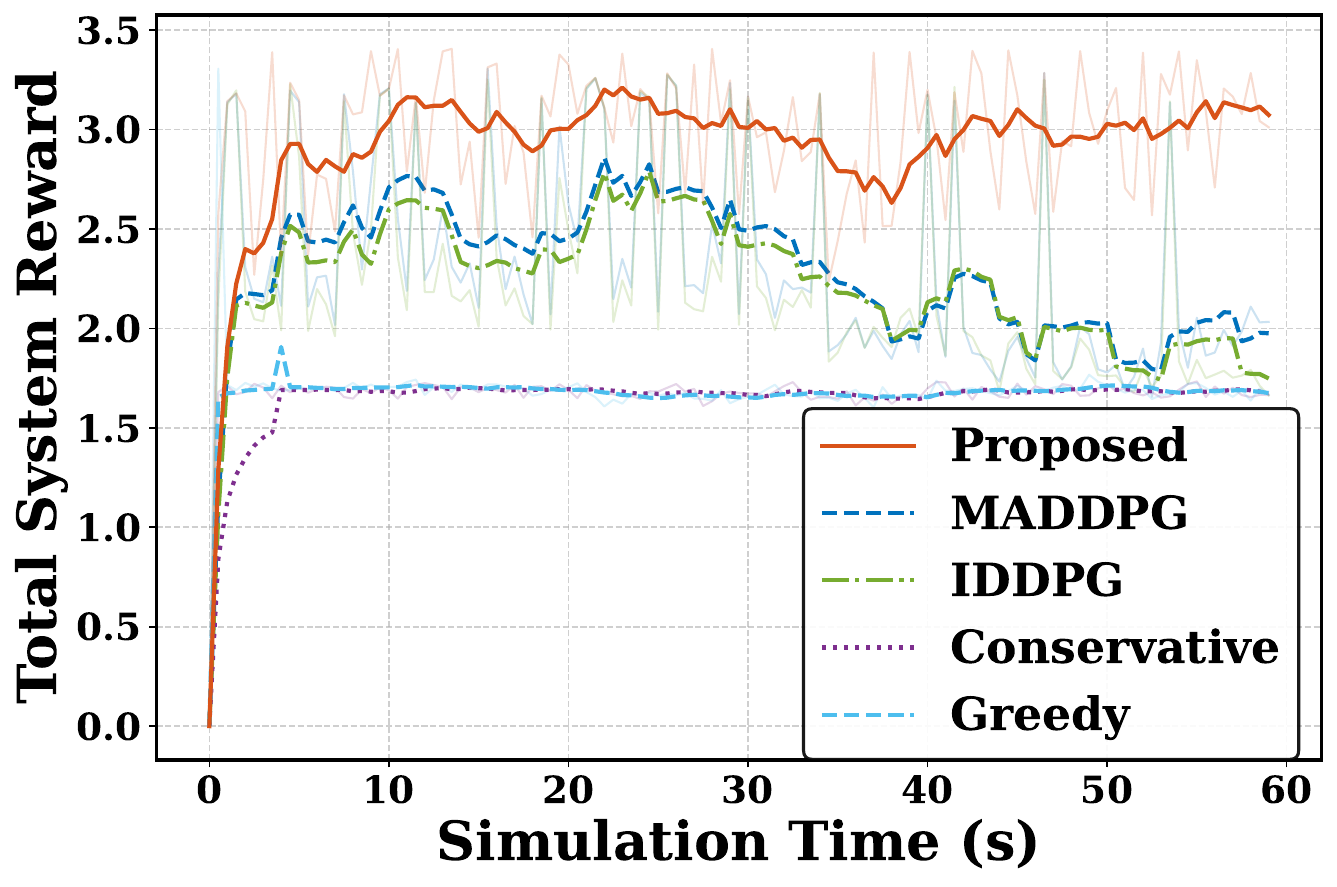}
        \caption{Total system reward.}
        \label{fig:reward_convergence}
    \end{subfigure}
    \begin{subfigure}[t]{.48\linewidth}
        \centering
        \includegraphics[width=\linewidth]{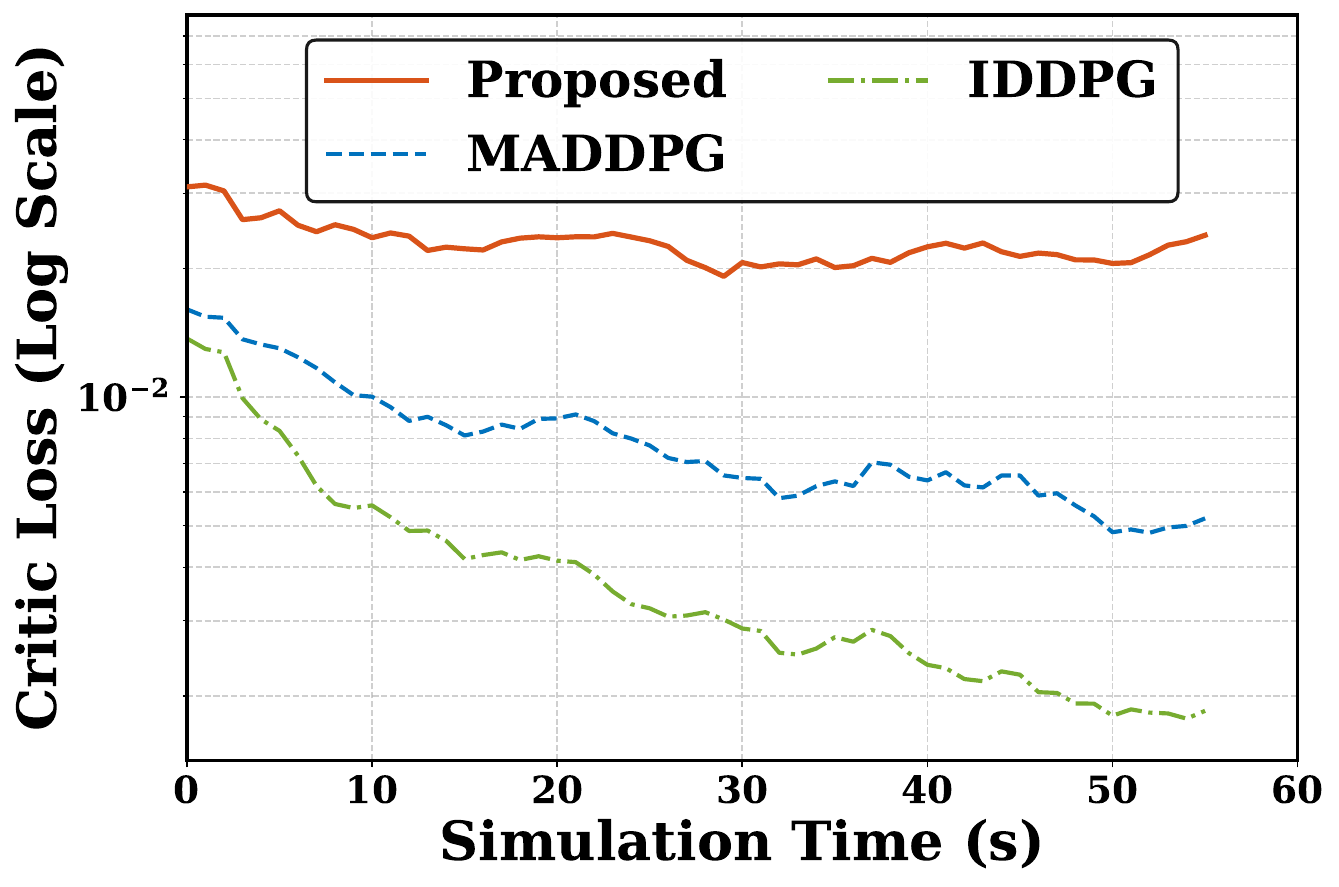}
        \caption{Critic loss in log scale.}
        \label{fig:critic_loss}
    \end{subfigure}
    \caption{Training convergence and critic stability.}
    \label{fig:convergence}
    \vspace{-0.5em}
\end{figure}

Fig.~\ref{fig:convergence} compares the training behavior of EvoOMG with representative baselines. 
As shown in Fig.~\ref{fig:convergence}(a), EvoOMG rapidly improves the system reward in the first few seconds and then stabilizes around $2.9$--$3.1$. 
MADDPG and IDDPG reach moderate rewards in the early stage but gradually decline in the later stage, mostly staying around $1.8$--$2.2$ after $40$~s. 
Conservative and Greedy policies remain nearly unchanged around $1.65$--$1.70$, showing limited adaptability to dynamic Wi-Fi MAC conditions.

Fig.~\ref{fig:convergence}(b) reports the critic loss in logarithmic scale. 
EvoOMG keeps a bounded critic loss, although its value is higher than those of MADDPG and IDDPG. 
This behavior is expected because EvoOMG uses a centralized critic to evaluate heterogeneous legacy/MLO interactions and autoregressive MAC actions. 
The lower critic losses of MADDPG and IDDPG do not lead to higher rewards, indicating that critic-loss minimization alone is insufficient; the critic must guide the actor toward protocol-aligned contention and aggregation decisions.

\subsection{Performance under Legacy and MLO Heterogeneity}

\begin{figure}[!t]
  \centering
  \includegraphics[width=\linewidth]{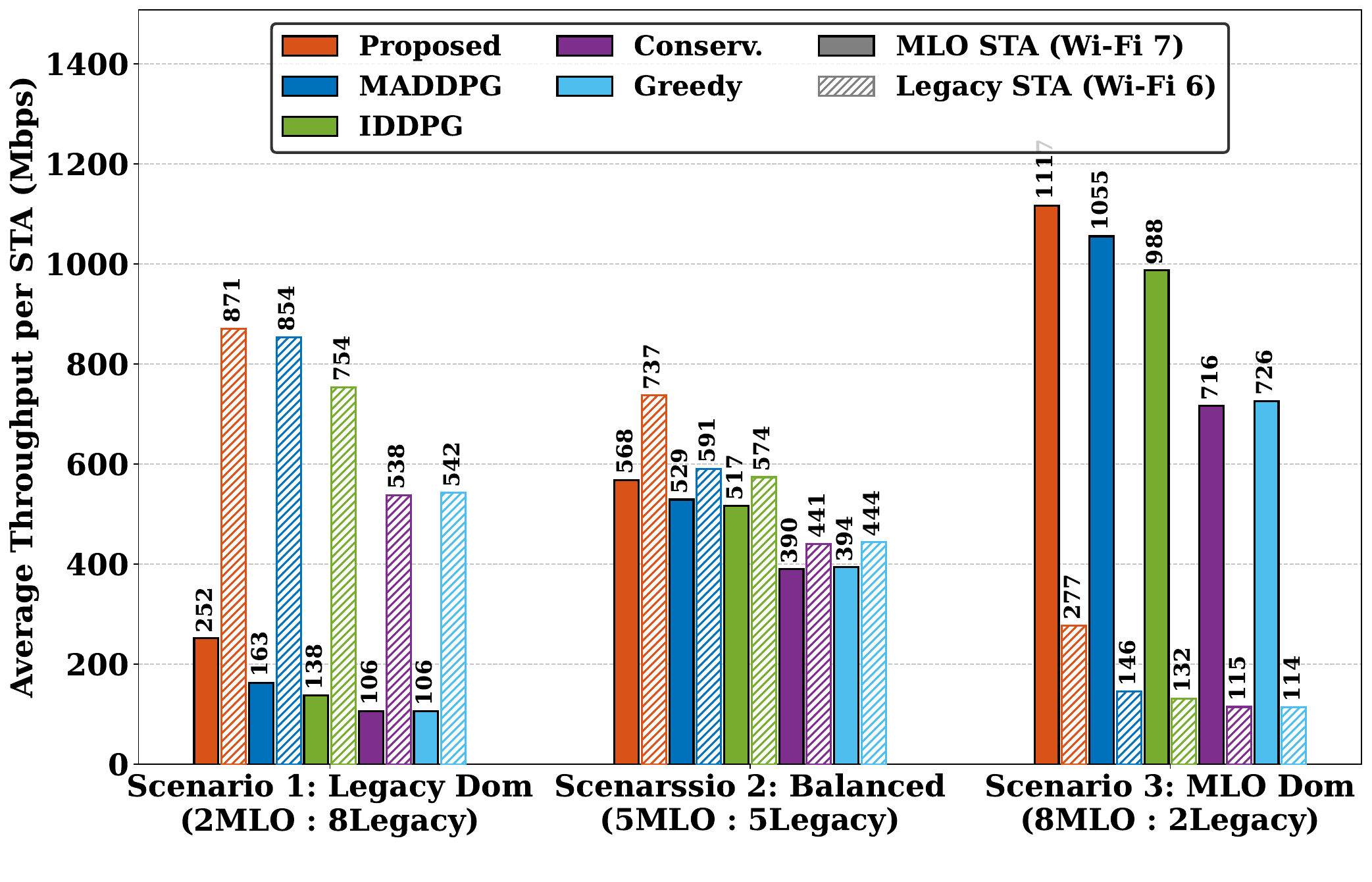}
 \caption{Throughput per STA under different legacy/MLO deployment ratios. Three settings are evaluated: legacy-dominant $(2\mathrm{MLO}:8\mathrm{Legacy})$, balanced $(5\mathrm{MLO}:5\mathrm{Legacy})$, and MLO-dominant $(8\mathrm{MLO}:2\mathrm{Legacy})$.}
  \label{fig:heterogeneity_fairness}
\end{figure}

Fig.~\ref{fig:heterogeneity_fairness} evaluates the impact of legacy/MLO device heterogeneity under three deployment ratios: legacy-dominant $(2\mathrm{MLO}:8\mathrm{Legacy})$, balanced $(5\mathrm{MLO}:5\mathrm{Legacy})$, and MLO-dominant $(8\mathrm{MLO}:2\mathrm{Legacy})$. 
Across all settings, EvoOMG achieves the highest or near-highest throughput for both MLO and legacy STAs, showing robust performance under mixed-standard coexistence.

In the legacy-dominant case, EvoOMG improves the MLO throughput from $163$ Mbps under MADDPG to $252$ Mbps, while keeping the legacy throughput slightly higher than MADDPG. 
In the balanced case, EvoOMG reaches $568$ Mbps and $737$ Mbps for MLO and legacy STAs, respectively, outperforming MADDPG and IDDPG on both device types. 
In the MLO-dominant case, EvoOMG further increases MLO throughput to $1111$ Mbps and significantly improves legacy throughput from $146$ Mbps under MADDPG to $277$ Mbps. 
These results indicate that EvoOMG can exploit MLO capability while preventing legacy STAs from being starved, owing to its autoregressive and standard-aware action generation.

\subsection{TID-Aware Traffic Differentiation and Link Scheduling}


\begin{figure}[!t]
    \centering
    \begin{subfigure}[t]{0.78\linewidth}
        \centering
        \includegraphics[width=\linewidth]{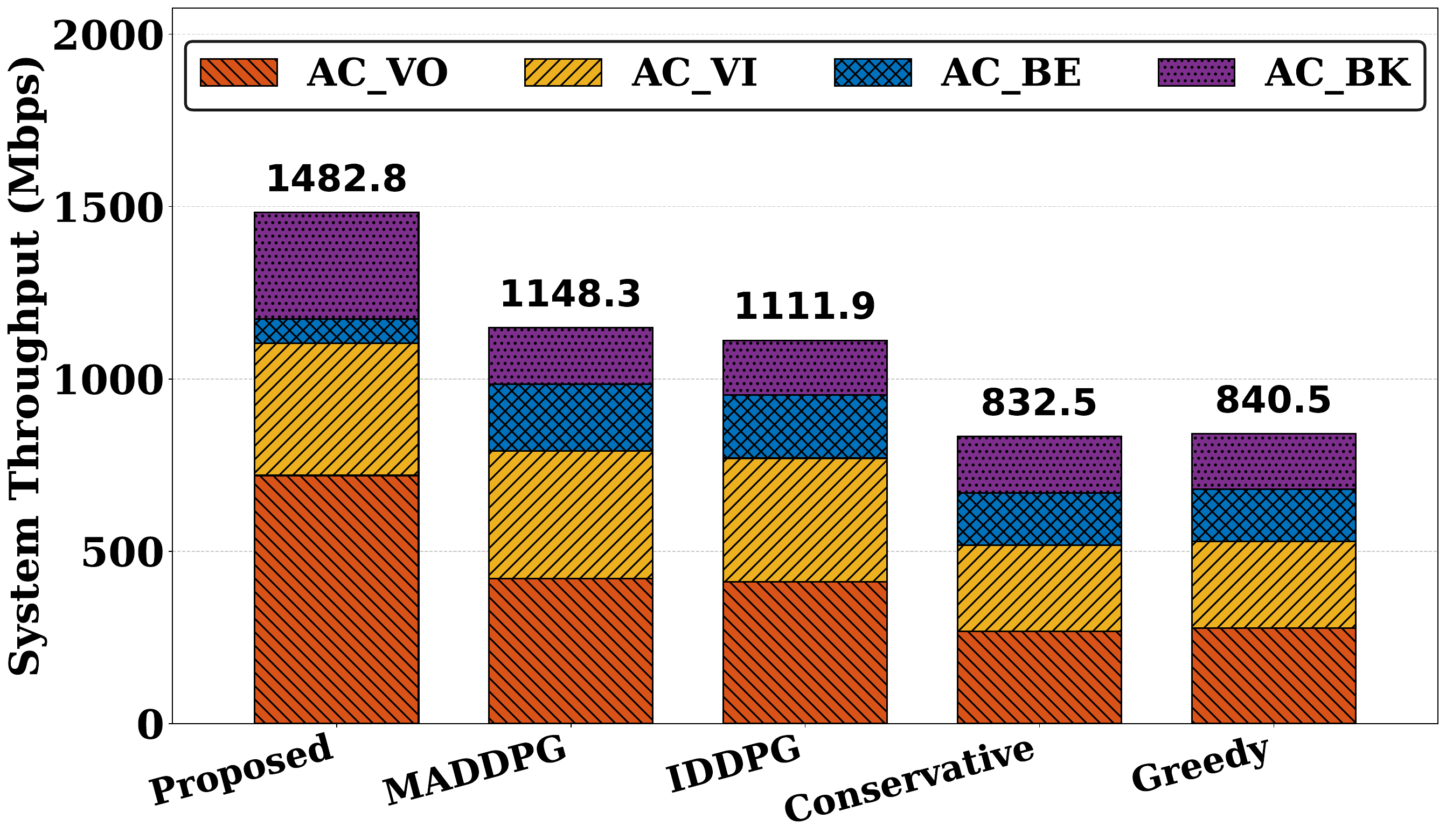}
        \caption{Stacked AC throughput.}
        \label{fig:qos_breakdown}
    \end{subfigure}
    \hfill
    \begin{subfigure}[t]{0.78\linewidth}
        \centering
        \includegraphics[width=\linewidth]{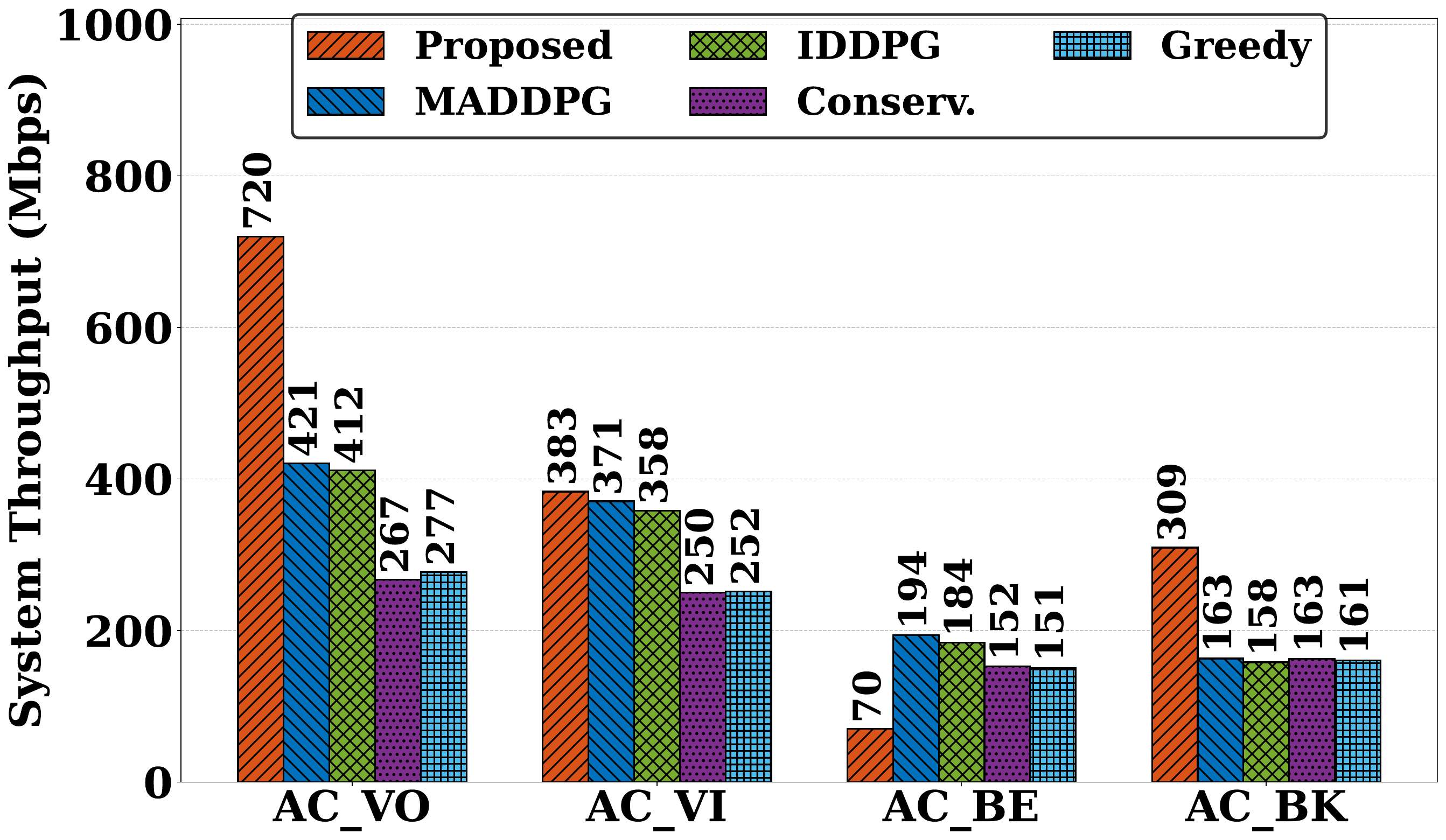}
        \caption{Per-AC throughput comparison.}
        \label{fig:qos_comparison}
    \end{subfigure}
    \caption{TID/AC-aware traffic differentiation results.}
    \label{fig:qos_ab}
    \vspace{-0.5em}
\end{figure}

\begin{figure}[!t]
    \centering
    \includegraphics[width=\linewidth]{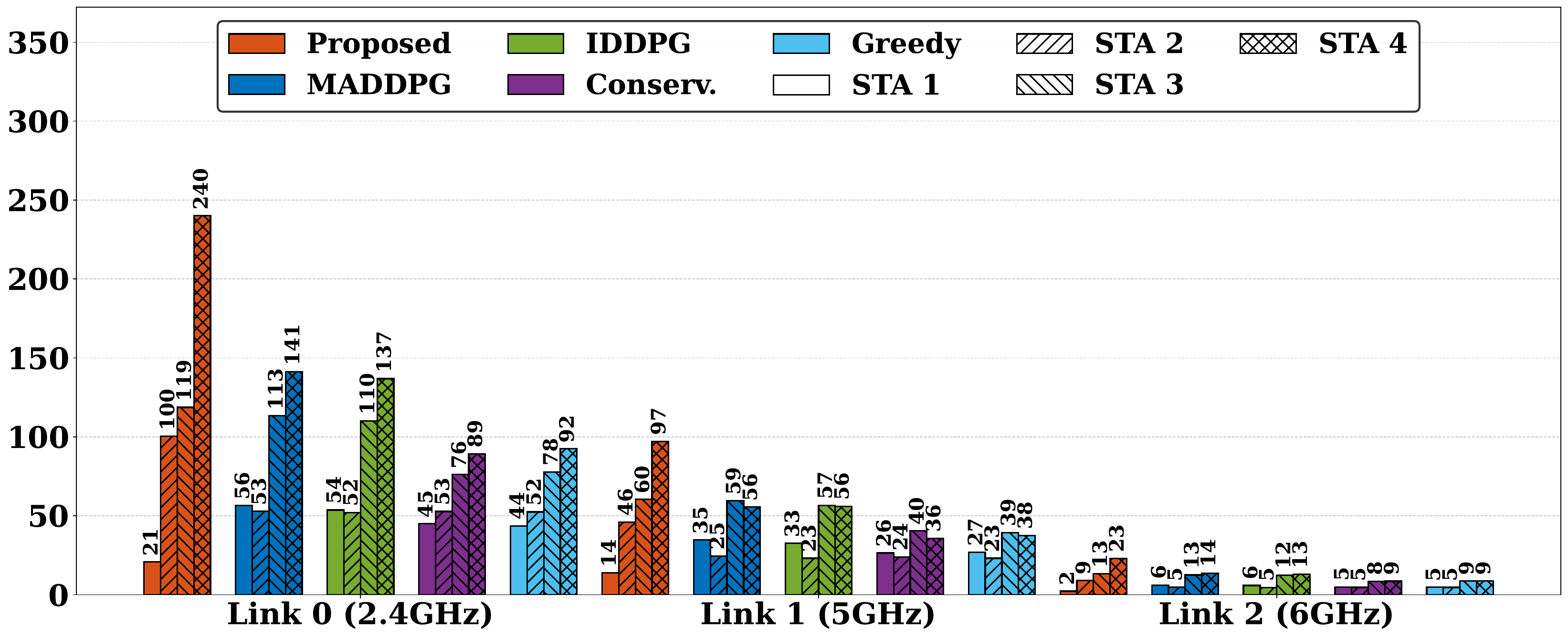}
    \caption{STA-level throughput breakdown over MLO physical links.}
    \label{fig:link_sta_breakdown}
    \vspace{-0.5em}
\end{figure}

Fig.~\ref{fig:qos_ab} evaluates TID/AC-aware scheduling and MLO link allocation. 
As shown in Fig.~\ref{fig:qos_ab}(a), EvoOMG achieves the highest total system throughput, reaching $1482.7$ Mbps, compared with $1148.3$ Mbps for MADDPG, $1111.9$ Mbps for IDDPG, $832.5$ Mbps for Conservative, and $840.5$ Mbps for Greedy. 
This corresponds to gains of $29.1\%$, $33.3\%$, $78.1\%$, and $76.4\%$, respectively.

Fig.~\ref{fig:qos_ab}(b) further shows the per-AC throughput. 
EvoOMG achieves the best performance on AC\_VO, AC\_VI, and AC\_BK, with $720$ Mbps, $383$ Mbps, and $309$ Mbps, respectively. 
For AC\_BE, EvoOMG obtains $70$ Mbps, lower than MADDPG and IDDPG, indicating that the learned policy prioritizes VO, VI, and BK traffic under the goodput-oriented utility. 
This reflects service differentiation rather than uniform allocation across all ACs. Fig.~\ref{fig:link_sta_breakdown} shows the STA-level traffic distribution over MLO links. 
EvoOMG mainly exploits Link~0 $(2.4~\mathrm{GHz})$, while still assigning part of the traffic to Link~1 $(5~\mathrm{GHz})$. 
Link~2 $(6~\mathrm{GHz})$ carries much less traffic, mostly below $20$ Mbps. 
This indicates that EvoOMG performs link-dependent scheduling instead of blindly activating all available MLO links.

\subsection{Packet-Level Tail Latency}

\begin{figure*}[!t]
    \centering
    \begin{subfigure}[t]{0.23\linewidth}
        \centering
        \includegraphics[width=\linewidth]{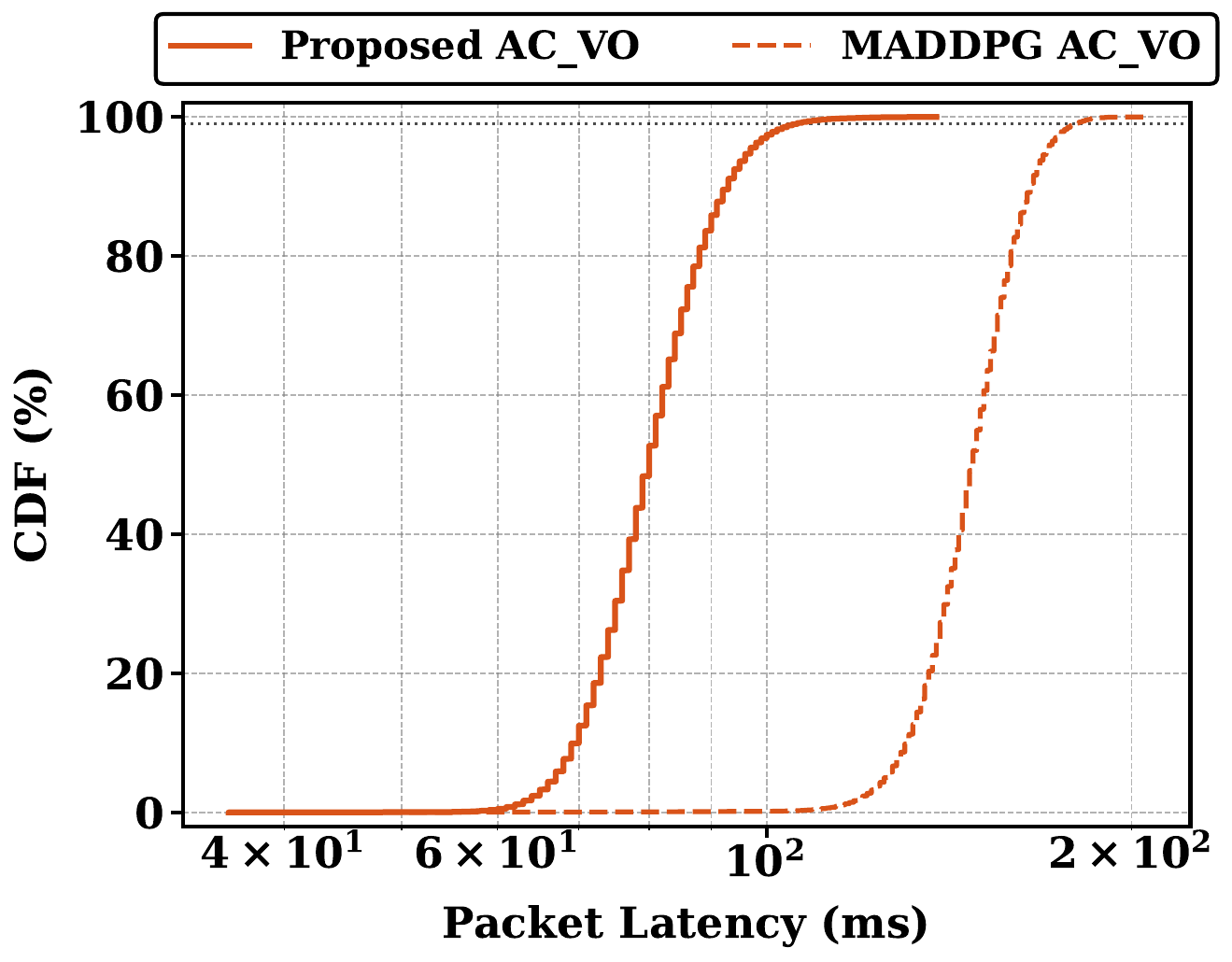}
        \caption{AC\_VO traffic.}
        \label{fig:tail_vo}
    \end{subfigure}
    \begin{subfigure}[t]{0.23\linewidth}
        \centering
        \includegraphics[width=\linewidth]{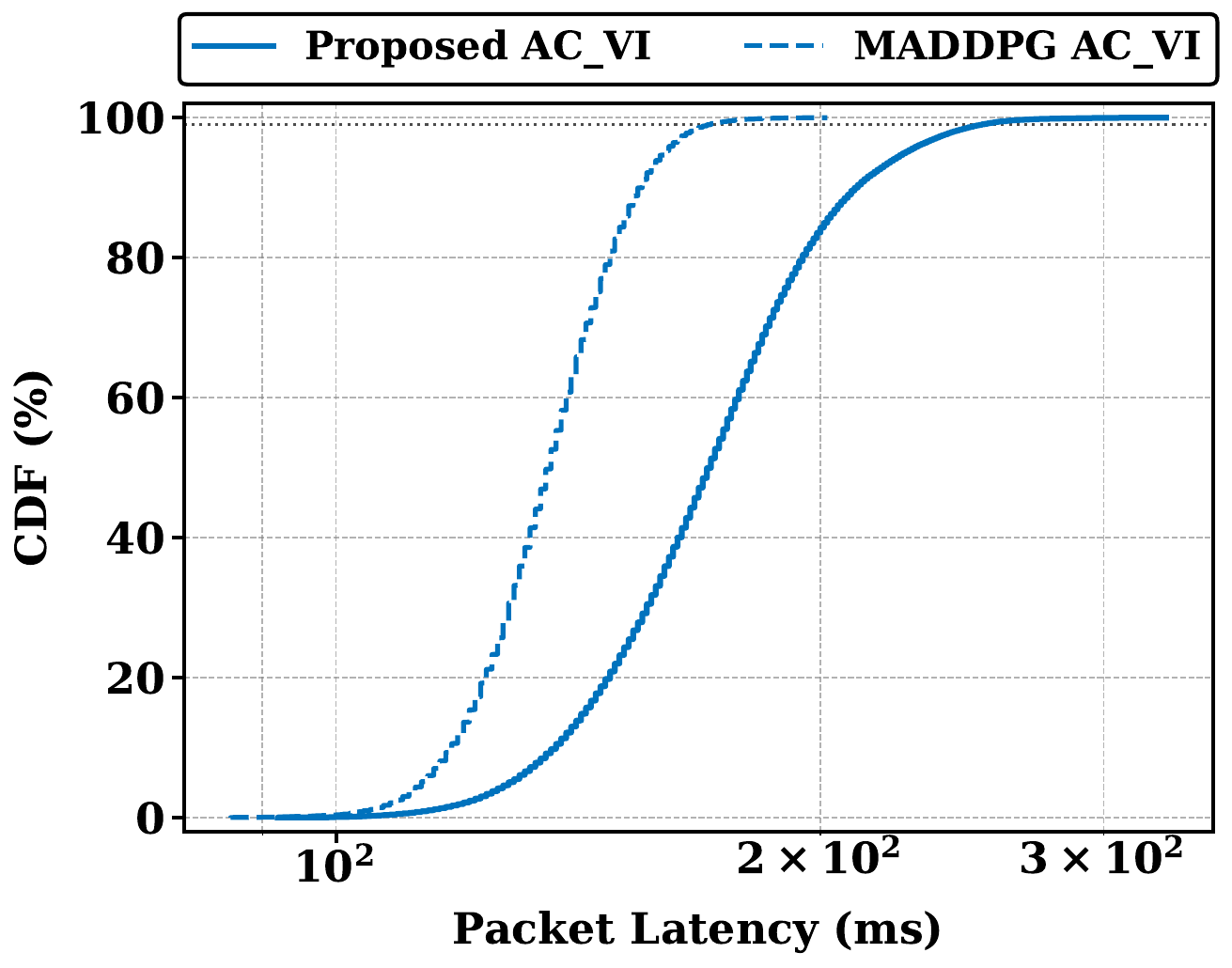}
        \caption{AC\_VI traffic.}
        \label{fig:tail_vi}
    \end{subfigure}
    \begin{subfigure}[t]{0.23\linewidth}
        \centering
        \includegraphics[width=\linewidth]{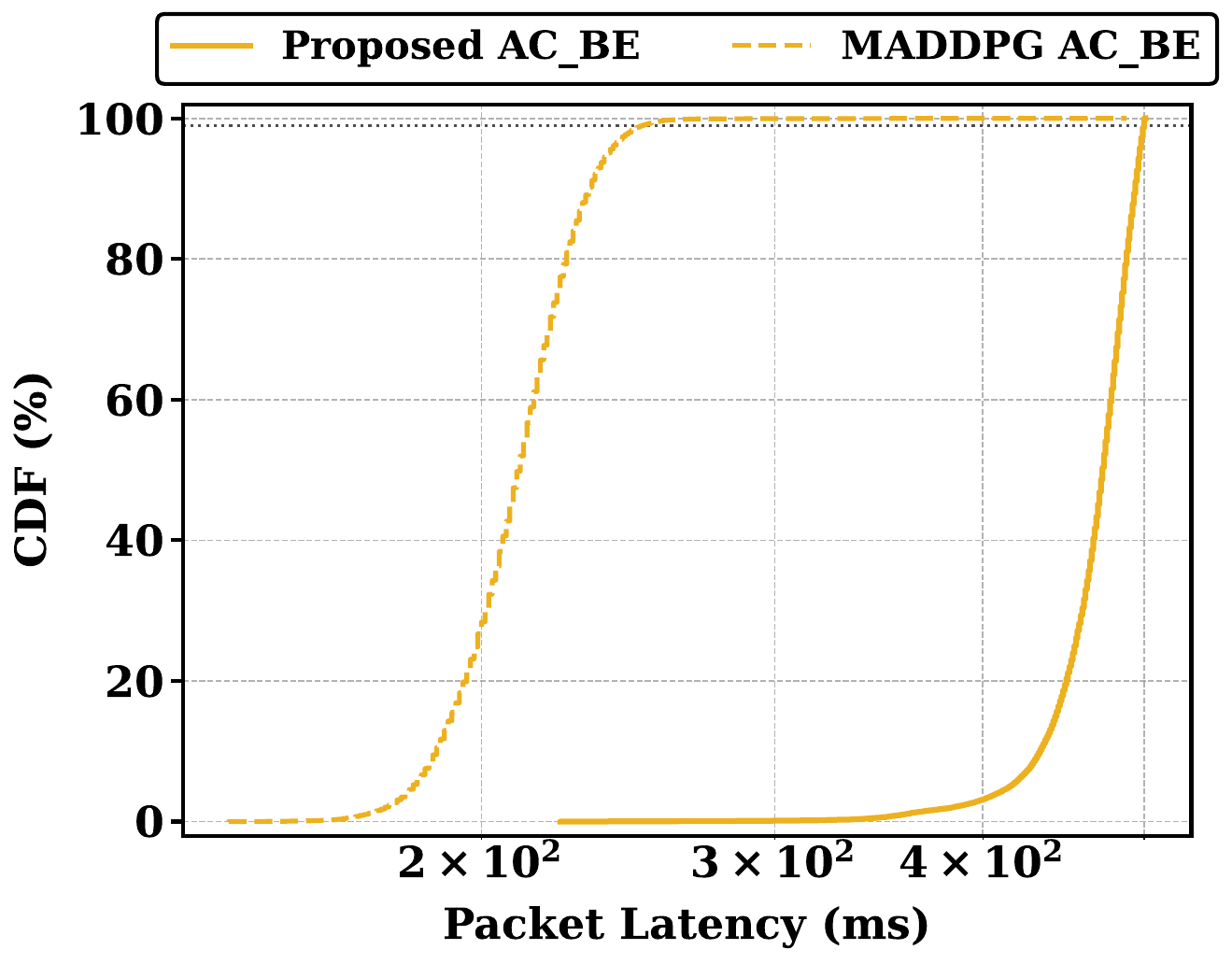}
        \caption{AC\_BE traffic.}
        \label{fig:tail_be}
    \end{subfigure}
    \begin{subfigure}[t]{0.23\linewidth}
        \centering
        \includegraphics[width=\linewidth]{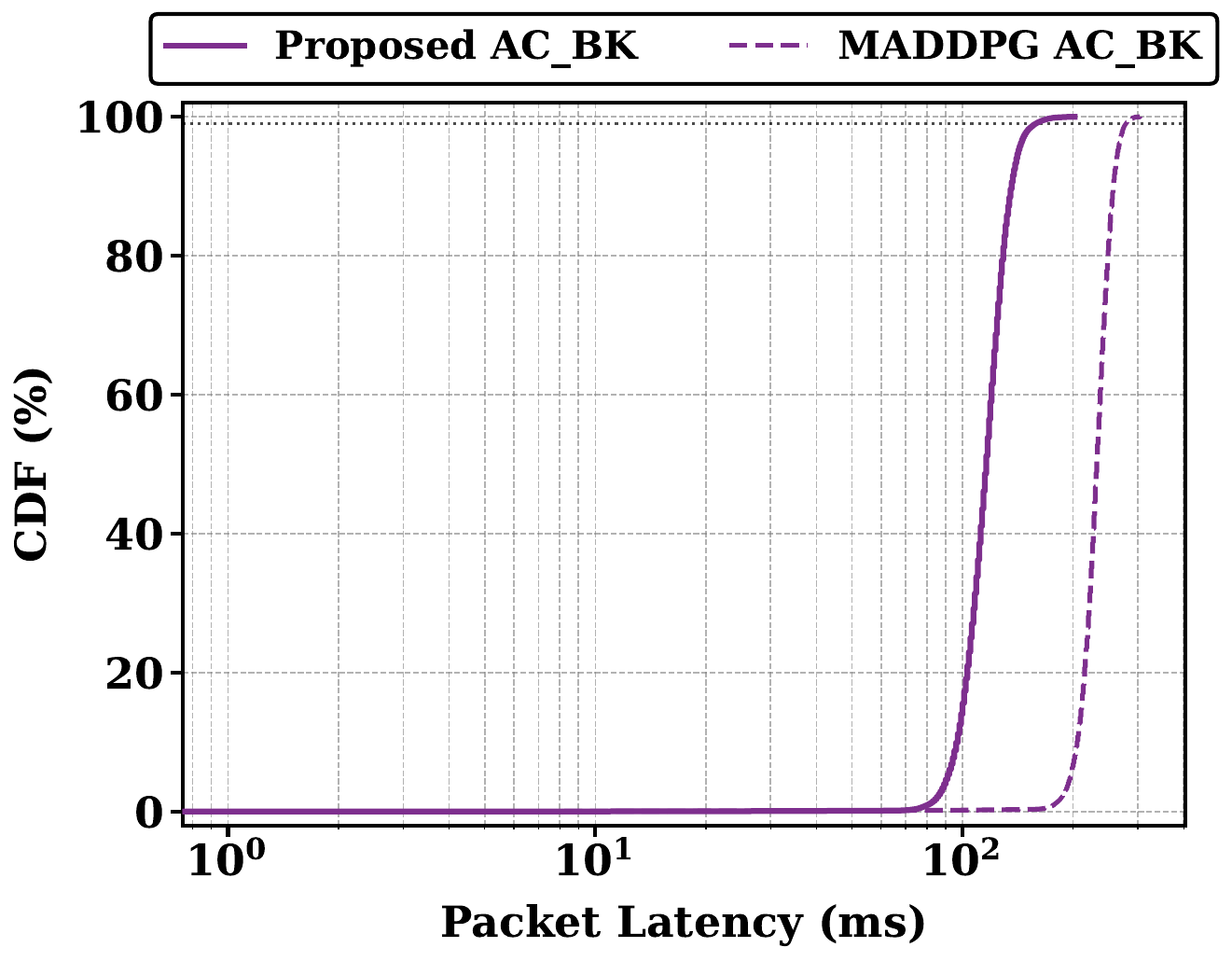}
        \caption{AC\_BK traffic.}
        \label{fig:tail_bk}
    \end{subfigure}
    \caption{Packet-level tail-latency CDFs for different AC traffic classes.}
    \label{fig:tail_latency}
    \vspace{-0.5em}
\end{figure*}

Fig.~\ref{fig:tail_latency} reports the packet-level latency cumulative distribution functions (CDFs) for different AC traffic classes. 
For AC\_VO in Fig.~\ref{fig:tail_latency}(a), EvoOMG significantly shifts the CDF to the left. 
Its median latency is around $75$--$80$ ms, while MADDPG is around $150$--$160$ ms. 
At the tail, EvoOMG approaches $100$--$110$ ms, whereas MADDPG is close to $190$ ms, showing a clear latency reduction for voice traffic.

For AC\_VI in Fig.~\ref{fig:tail_latency}(b), MADDPG has a slightly lower delay distribution than EvoOMG. 
The median latency of EvoOMG is around $180$ ms, while MADDPG is around $145$--$150$ ms. 
This indicates that EvoOMG does not uniformly minimize delay for all ACs; instead, it trades part of the VI delay for higher system throughput and better global scheduling.

For AC\_BE in Fig.~\ref{fig:tail_latency}(c), EvoOMG has a much larger latency than MADDPG. 
The EvoOMG CDF rises mainly around $430$--$460$ ms, while MADDPG is concentrated around $200$--$260$ ms. 
This is consistent with the per-AC throughput result, where BE traffic receives fewer transmission opportunities under the goodput-oriented utility.

For AC\_BK in Fig.~\ref{fig:tail_latency}(d), EvoOMG shifts the CDF to the left compared with MADDPG. 
Most BK packets under EvoOMG are completed around $100$--$130$ ms, while MADDPG is closer to $220$--$260$ ms. 
Overall, EvoOMG substantially improves VO and BK latency, while VI and BE reflect the trade-off between goodput maximization and AC-level delay fairness.

\subsection{Performance under STR/NSTR PHY Constraints}

\begin{figure*}[!t]
    \centering
    \begin{subfigure}[t]{0.23\linewidth}
        \centering
        \includegraphics[width=\linewidth]{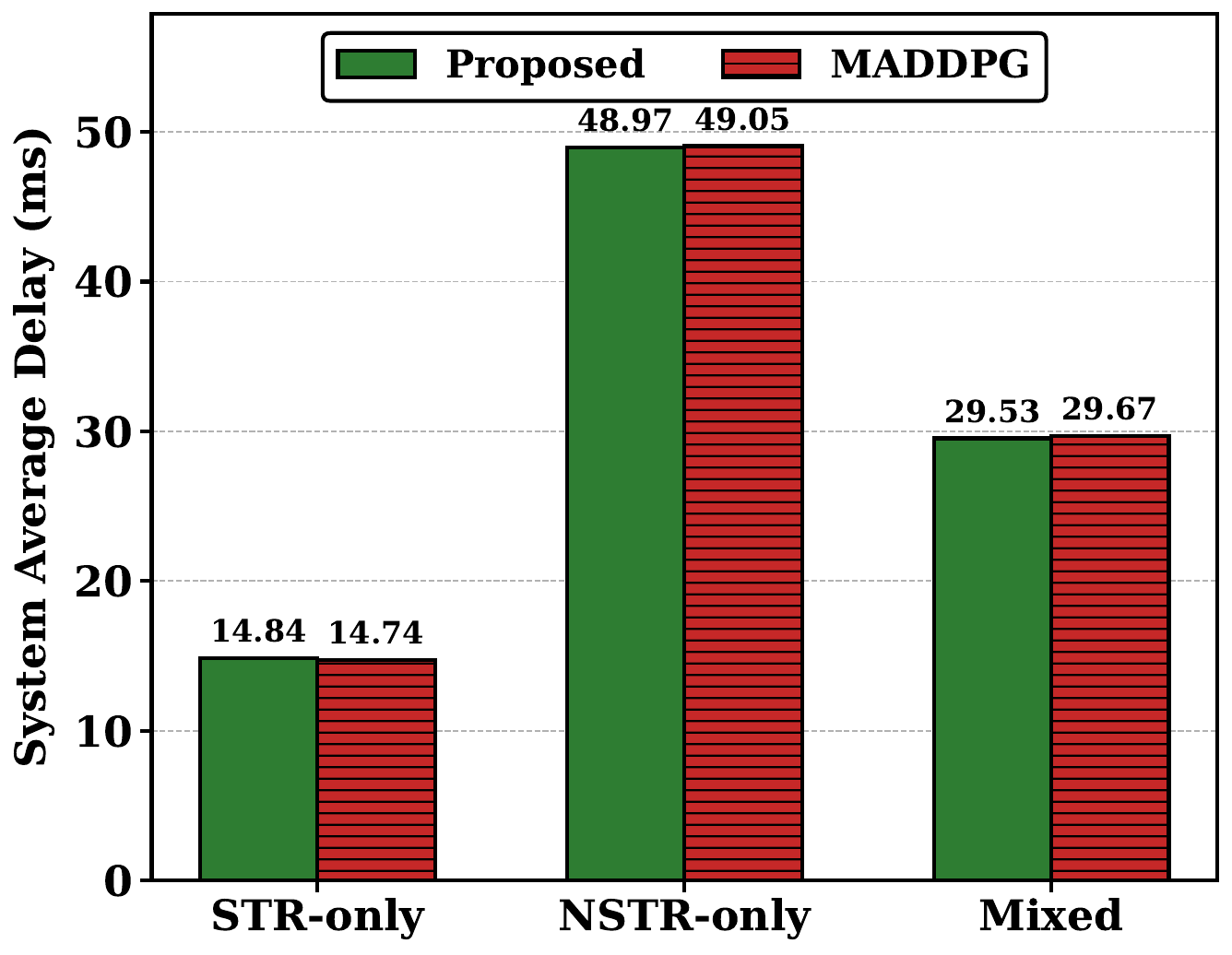}
        \caption{Average delay.}
        \label{fig:phy_delay}
    \end{subfigure}
    \begin{subfigure}[t]{0.23\linewidth}
        \centering
        \includegraphics[width=\linewidth]{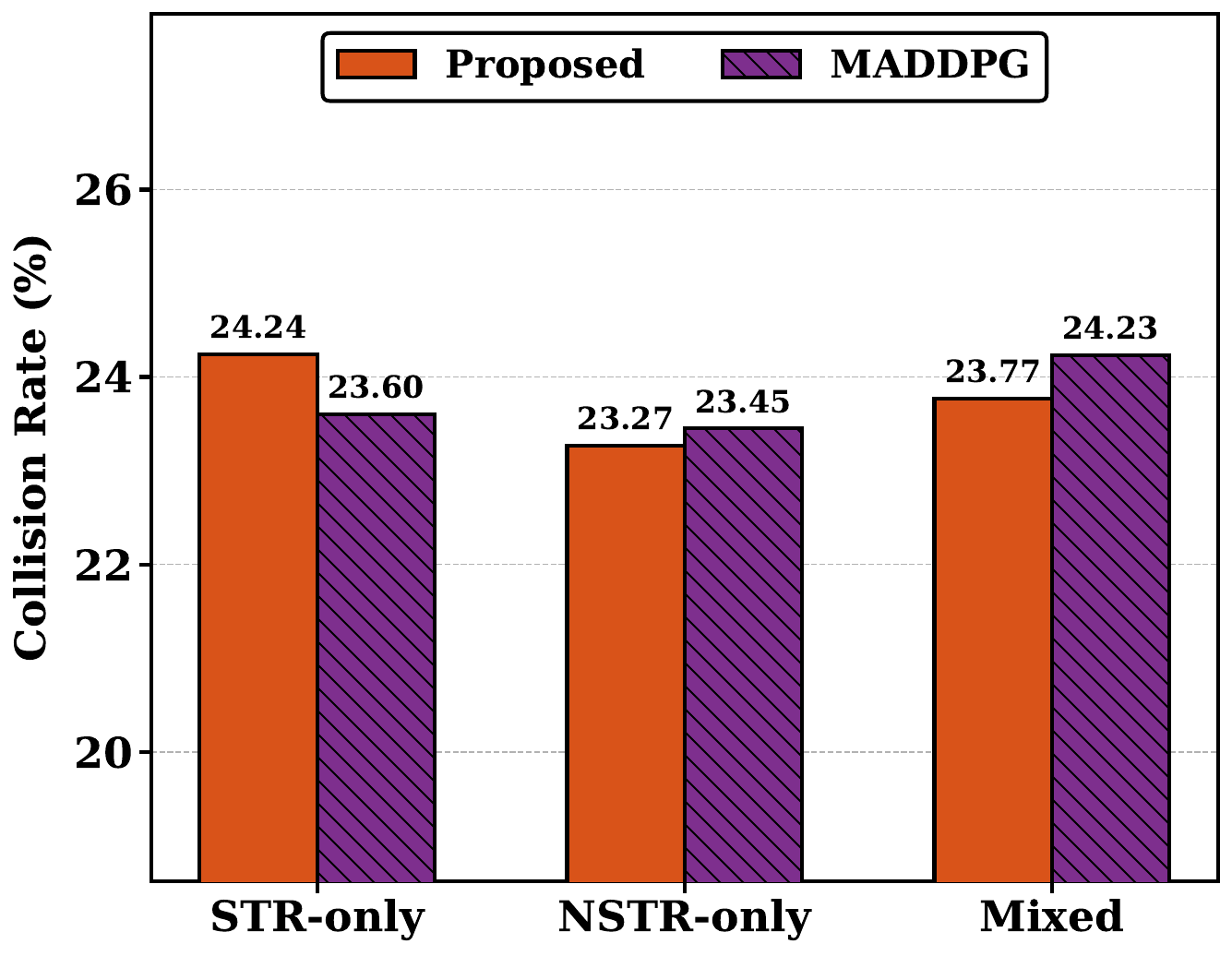}
        \caption{Collision rate.}
        \label{fig:phy_collision}
    \end{subfigure}
    \begin{subfigure}[t]{0.23\linewidth}
        \centering
        \includegraphics[width=\linewidth]{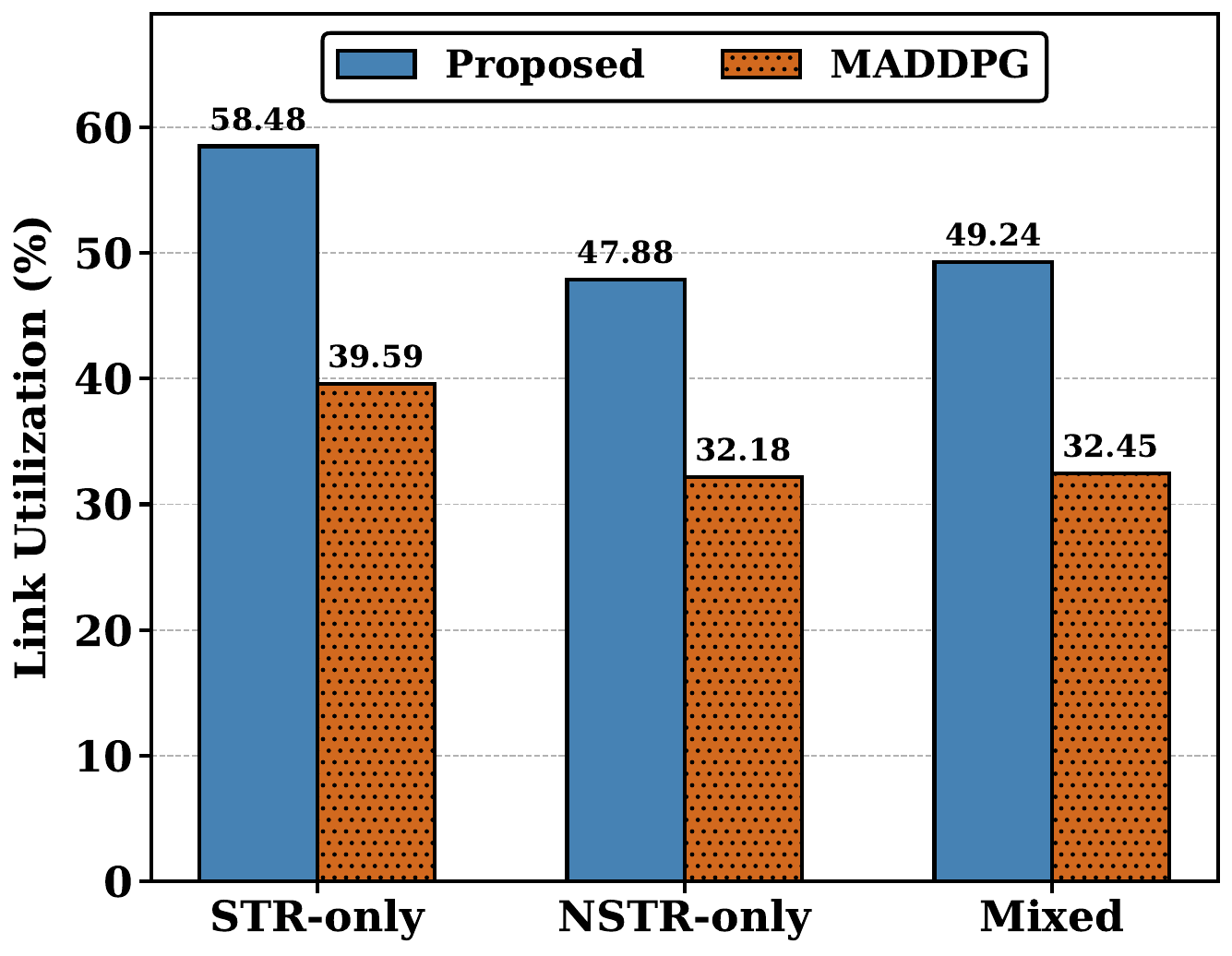}
        \caption{Link utilization.}
        \label{fig:phy_utilization}
    \end{subfigure}
    \begin{subfigure}[t]{0.23\linewidth}
        \centering
        \includegraphics[width=\linewidth]{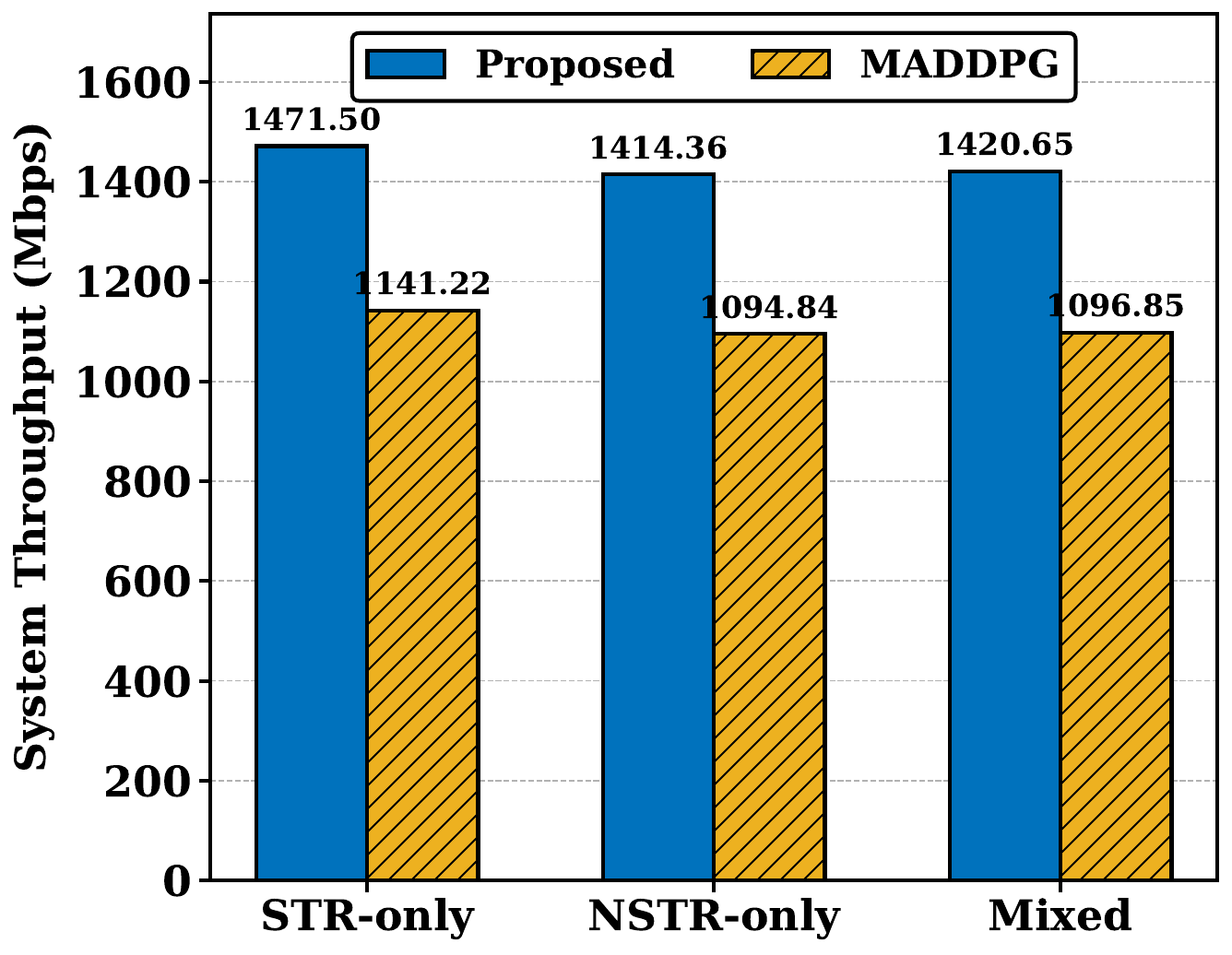}
        \caption{System throughput.}
        \label{fig:phy_throughput}
    \end{subfigure}
    \caption{Performance under different MLO PHY constraints. EvoOMG is compared with MADDPG under STR-only, NSTR-only, and mixed STR/NSTR settings.}
    \label{fig:phy_constraints}
    \vspace{-0.5em}
\end{figure*}

Fig.~\ref{fig:phy_constraints} evaluates EvoOMG under STR-only, NSTR-only, and mixed STR/NSTR settings. 
As shown in Fig.~\ref{fig:phy_constraints}(a), EvoOMG and MADDPG have almost identical average delay across the three modes. 
The delays of EvoOMG are $14.84$ ms, $48.97$ ms, and $29.53$ ms under STR-only, NSTR-only, and mixed settings, respectively, which are very close to MADDPG's $14.74$ ms, $49.05$ ms, and $29.67$ ms. 
This indicates that EvoOMG improves throughput without increasing average delay.

Fig.~\ref{fig:phy_constraints}(b) shows that the collision rates are also comparable. 
EvoOMG obtains $24.24\%$, $23.27\%$, and $23.77\%$ under the three settings, while MADDPG obtains $23.60\%$, $23.45\%$, and $24.23\%$. 
Thus, the performance gain of EvoOMG is not mainly caused by collision reduction.

The main advantage appears in Fig.~\ref{fig:phy_constraints}(c). 
EvoOMG achieves link utilization of $58.48\%$, $47.88\%$, and $49.24\%$, clearly higher than MADDPG's $39.59\%$, $32.18\%$, and $32.45\%$. 
This shows that the proposed policy can better exploit feasible MLO links under both STR and NSTR constraints.

Fig.~\ref{fig:phy_constraints}(d) further confirms this advantage. 
EvoOMG reaches $1471.50$ Mbps, $1414.36$ Mbps, and $1420.65$ Mbps under STR-only, NSTR-only, and mixed settings, respectively, while MADDPG achieves $1141.22$ Mbps, $1094.84$ Mbps, and $1096.85$ Mbps. 
The corresponding throughput gains are about $28.9\%$, $29.2\%$, and $29.5\%$. 
These results demonstrate that EvoOMG remains effective when full simultaneous transmission is available and when NSTR constraints restrict feasible link activation. The results show that EvoOMG consistently approaches the STR-enabled performance upper envelope even under NSTR constraints.

\subsection{Ablation Study}

\begin{figure}[!t]
  \centering
  \includegraphics[width=\linewidth]{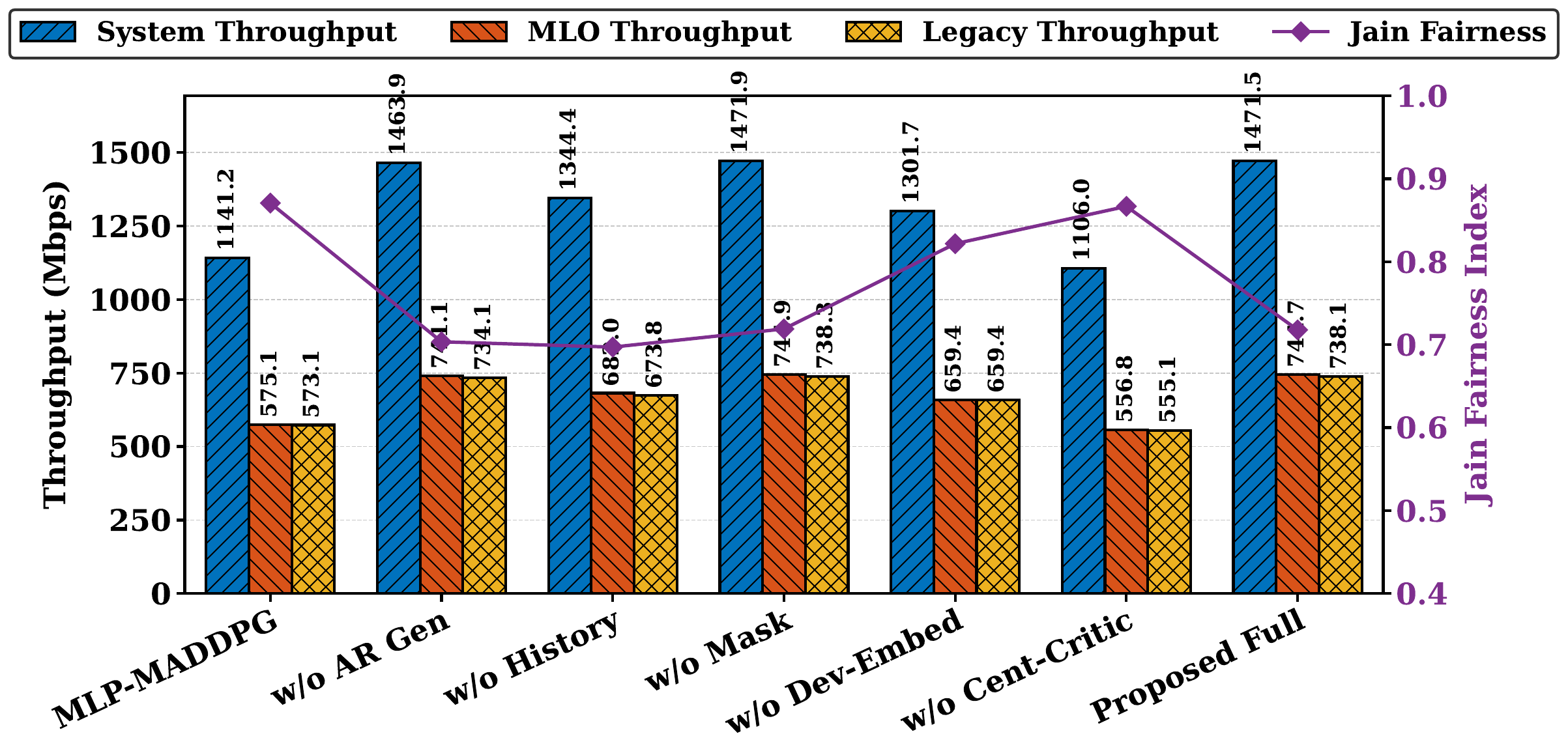}
  \caption{Ablation study of EvoOMG core mechanisms. }
  \label{fig:ablation}
  \vspace{-0.5em}
\end{figure}

Fig.~\ref{fig:ablation} evaluates the contribution of each component using the following variants: \begin{itemize}[leftmargin=1.2em] \item \textbf{MLP-MADDPG}: a one-shot MLP-based MADDPG baseline without Transformer encoding or autoregressive action generation. \item \textbf{w/o AR Gen}: EvoOMG without autoregressive generation, where CW and A-MPDU length are generated as a flat joint action. \item \textbf{w/o History}: EvoOMG without historical state encoding, using only the current observation. \item \textbf{w/o Mask}: EvoOMG without standard-aware feasibility masking for legacy/MLO action constraints. \item \textbf{w/o Dev-Embed}: EvoOMG without the device-standard embedding that distinguishes legacy and MLO STAs. \item \textbf{w/o Cent-Critic}: EvoOMG without the centralized critic, relying on decentralized value estimation. \item \textbf{Proposed Full}: the complete EvoOMG model with Transformer-based history encoding, autoregressive action generation, device embedding, feasibility masking, and centralized critic. \end{itemize}

Proposed Full achieves the highest system throughput of $1471.5$ Mbps, with $744.7$ Mbps MLO throughput and $738.1$ Mbps legacy throughput. 
Compared with MLP-MADDPG, whose system throughput is $1141.2$ Mbps, Proposed Full improves system throughput by about $28.9\%$. The w/o AR Gen variant reaches $1463.9$ Mbps, close to Proposed Full, but its Jain fairness is lower, indicating that autoregressive CW-to-aggregation generation mainly improves coordinated allocation rather than only increasing raw throughput. 
The w/o History variant drops to $1344.4$ Mbps, showing that recent channel, queue, contention, and aggregation states are useful for MAC decision making. 
The w/o Mask variant obtains $1471.9$ Mbps, but its fairness is also lower than Proposed Full, suggesting that standard-aware masking helps maintain a better legacy/MLO balance. 
The w/o Dev-Embed variant decreases to $1301.7$ Mbps, confirming the importance of explicitly encoding device type for legacy and MLO action grammars. 
Finally, w/o Cent-Critic performs the worst with $1106.0$ Mbps, demonstrating that centralized training is critical for learning coupled multi-agent contention behavior.

\section{Conclusion}
This paper argues that Wi-Fi optimization should be revisited from two first principles: standard heterogeneity and protocol sequentiality. Mixed deployments of legacy non-MLO STAs and MLO-capable STAs invalidate one-size-fits-all control, while the contention-to-transmission pipeline makes one-step static policies structurally mismatched to the MAC process. To address these issues, we propose EvoOMG, an evolution-oriented multi-agent guidance framework whose internal policy is sequence-aware and stage-aligned with Wi-Fi access. The draft provides the system model, algorithm design, theoretical interpretation, and evaluation plan for a full NS-3 study. Our core thesis is that protocol-aligned staged guidance is a promising direction for adaptive intelligence in heterogeneous legacy-and-MLO Wi-Fi networks.

\balance
\bibliographystyle{IEEEtran}
\bibliography{sample-base1}








\end{document}